\documentclass[aps,prl,twocolumn,showpacs,superscriptaddress,groupedaddress]{revtex4}  
\usepackage{graphicx}  
\usepackage{dcolumn}   
\usepackage{pictex}

\usepackage{bm}        
\usepackage{amsfonts,amssymb,amsthm,amsmath}   
\usepackage{threeparttable}

\usepackage[caption=false]{subfig}
\usepackage{verbatim}
\newtheorem{theorem}{Theorem}

\usepackage[usenames,dvipsnames,svgnames]{xcolor}

\def\2;{\;\;}

\def\Ref#1{(\ref{#1})}

\def\C#1{{\mathcal #1}}

\begin{document}

\widetext
\leftline{Version 1.0 as of \today}

\title[Thermodynamics of Confined Knotted Lattice Polygons]{
Thermodynamics of Confined Knotted Lattice Polygons}
\author{EJ Janse van Rensburg}
\affiliation{Department of Mathematics and Statistics, 
York University, Toronto, Ontario M3J~1P3, Canada 
(Email: rensburg@yorku.ca)\\}
\author{E Orlandini}
\affiliation{Dipartimento di Fisica e Astronomia ``Galileo Galilei'', 
Universit\'a degli studi di Padova, 8 - 35131 Padova, Italia
(Email: orlandini@pd.infn.it)\\}
\author{MC Tesi}
\affiliation{Dipartimento di Matematica, Universit\`a di Bologna,
Piazza di Porta San Donato 5, I-40126 Bologna, Italy
(Email: mariacarla.tesi@unibo.it)\\}

\date{\today}

\begin{abstract}
A ring polymer in a confining space may exhibit at least two phases,
namely an expanded (or solvent-rich phase) if its concentration is small,
or a collapsed (or polymer-rich phase) when it is concentrated and compressed.
These phases are discussed in reference \cite{deG79}, and have been
modelled, traditionally, in the mean field using Flory-Huggins 
theory \cite{Flory42,Huggins42}.  In three 
dimensions the ring polymer may also be knotted, or linked, and have
its conformational degrees of freedom constrained by its topology.  In a lattice
model of confined knotted ring polymers there are indications that
the thermodynamic properties of the ring polymer (for example, 
the osmotic pressure \cite{GJvR18,JvR19}) is a function of its topology.  In this paper
we explore a lattice knot model of a confined ring polymer as a function of its
chemical potential.  We show that a well-defined phase transition occurs 
between solvent-rich and polymer-rich phases when the lattice knot exhibits 
either the unknot topology or any other fixed knot type. Furthermore, 
we observe small yet significant variations in the free energy near the 
critical point when comparing trefoil knots with other non-trivial knot types. 
These findings indicate that the thermodynamic properties of confined 
ring polymers depend on their topological entanglement 
characteristics (namely, their knot type).
\end{abstract}

\pacs{82.35.Lr,\,82.35.Gh,\,61.25.Hq}
\maketitle


\section{Introduction}

It is well established in polymer physics and biophysics that the occurrence 
and dimension of knots in circular fluctuating chains are severely affected 
by their lengths and physical properties,  environmental conditions, and 
externally imposed constraints \cite{OW07,MMO07}. 
For instance, the probability of knotting is enhanced if the rings either 
undergo a collapse ($\theta$-point) transition 
\cite{TJvROSW94,OSV04} or 
are isotropically confined \cite{MMOS06,MMOS08,MLY11,MOSSTM09,TOM11a,JOT25}. 
In contrast, a decrease in knotting probability is observed when the ring 
is confined in channels and slabs or stretched by a force 
\cite{TJOW94,JOTW07,JOTW07b,MO12,MO12a}. 
A peculiar behaviour is also observed in  semiflexible polymers, 
where the knotting probability displays a non-monotonic dependence on 
the bending rigidity \cite{COM17}.

Entanglements in polymer strands can be characterized by the degree of
knotting and linking of segments along the polymer strands.  Locating the
``knotted portion'' of a polymer is well defined for small tight knots, 
but is difficult when knots are strongly geometrically entangled, as in 
very long or highly compact chains.  Direct measurements of the size of
a knotted segment, based on various closure schemes, have been proposed 
in recent years \cite{MOSZ05,MDS05,TOM11}. Using 
these approaches, it has been shown that, in good solvents, knotted segments 
along polymers are  weakly localized (that is, they grow sublinearly with 
the chain length).  In a compressed or globule phase the segments tend to
be delocalized, spreading throughout the entire chain \cite{TOM11a}.

An indirect measure of the size of knotted segments along a ring polymer
is based on the scaling behaviour of the configurational entropy of knotted 
rings \cite{OTJW96}. This approach has yielded indications 
that fit well with the findings based on the direct methods, especially in 
good solvent and unconstrained conditions 
\cite{OTJW96,OTJW98,BOS10} 
and partially in the collapsed phase \cite{MOSZ07,BO14}.

In the presence of self and mutual entanglements, polymers are known to 
respond in a non-trivial way to thermal fluctuations and external perturbations 
\cite{SW93,V95}. This topological dependence was examined for ring 
polymers, for example, in reference \cite{GFR96} using mean field 
Flory arguments. Vanderzande discussed the thermodynamics of 
adsorbing ring polymers of non-trivial fixed knot type in reference 
\cite{V95}.  Similarly, the topological effects in ring polymers were 
numerically examined in reference \cite{MWC96} in a model of a melt of 
non-catenated ring polymers; see also reference \cite{HKG13}.

Models of ring polymers were introduced more than 70 years ago 
\cite{ZS49}, and a sound mathematical basis for lattice polygon models 
of ring polymer free energy was introduced in references \cite{H61A,MS93}.  
In particular, the existence of a connective constant made possible proofs 
of the existence of thermodynamic limits in these models, including, 
in many cases, lattice models of interacting or self-interacting ring 
polymers, see, for example, reference \cite{JvR15}.  More generally, 
various lattice and other models of knotting and entanglements in 
ring polymers and biopolymers have been studied 
\cite{D62,deG84,KM91,RCV93,SW93,GFR96,BO12}, particularly in 
the cubic lattice \cite{MW84,JvRW91,KM91,JvR02}.  These include
models of polymers in confining spaces \cite{MLY11,GJR12}, or 
adsorbed polymers \cite{V95}, or a polymer in the collapsed or 
dense phase~\cite{TJvROSW94,OSV04} (modelling a single ring in 
a ring polymer melt).  In references \cite{OJTW94,JOTW22} lattice 
models of the catenation of two ring polymers are examined, also 
as the $\theta$-point is crossed where there is an increase in linking probability.

An important yet not fully explored facet of this research area concerns the 
thermodynamics of knotted ring polymers confined within cavities. Previous 
studies have shown that thermodynamic properties such as pressure 
\cite{JvR07} and osmotic pressure \cite{JvR19,G20,JvR20}—exhibit non-trivial 
dependencies on knotting. Similar results have been observed in lattice 
models of compressed linear polymers \cite{GJvR18}.
If the length of a ring polymer increases in a cavity of fixed size it undergoes
a collapse transition from a solvent rich phase to a dense or polymer rich phase
\cite{deG79}.  The degree of entanglement of a polymer increases through the
transition to the dense phase, including a dramatic increase in the probability 
of knotting.  This has been partially verified for an off-lattice model of ring
polymers in spheres \cite{MMOS06,MMOS08,MLY11,MOSSTM09,TOM11a,CDSSU26}, 
and in space-filling melts of self-assembled polygons 
\cite{SM25}.  However, in these regimes some questions 
are still open: (i) If the knot type of a ring polymer in a cavity is fixed, then
how does the thermodynamics of the solvent-polymer phase transition 
depend on knot type? (ii) Is this transition accompanied by a delocalisation 
transition (or melting) of the knotted segment in the ring polymer? 

In this paper we partially address these questions by performing a 
systematic, extensive Monte Carlo simulation of lattice knots confined
in a cubic cavity in the grand canonical ensemble.  That is, knotted
lattice polygons in a cube with their length varying according to a 
chemical potential.

\section{The model}

A \textit{lattice polygon} is an embedding of a closed curve in the lattice
realised by a sequence of $n$ consecutive unit-length steps and $n$ 
distinct vertices with the first and last vertices being the same.  A lattice
polygon is \textit{confined} if it is placed within a cube of side-length 
$L-1$ and volume $V=L^3$ lattice sites.  If it is knotted with a fixed 
knot type $K$, then it is a  \textit{lattice knot}.  An example of a confined
lattice knot is shown in figure \ref{3}.   A \textit{lattice unknot} is a 
lattice knot of knot type the unknot, denoted by $K=0_1$.   The simplest
non-trivial lattice knot has knot type denoted by $3_1$, and it is a 
lattice trefoil.  The trefoil knot type is chiral, meaning that it is not equivalent
to its mirror reflection.  The right-handed version is denoted by $3_1^+$,
and the left-handed version by $3_1^-$.  Knot types are classified by
their minimal crossing number (see, for example, reference \cite{R03}), and in
this paper we shall consider lattice knots of knot types up to $6$
crossing knots, including the compound knot types $3_1^\pm \# 3_1^\pm$
obtained by the direct sum of trefoil knots.  For details, see reference
\cite{R03}.  In figure \ref{3} a confined lattice knot of knot type 
$5_1^+$ is shown (it is a chiral knot type, similar to the trefoil).

\begin{figure}[h!]
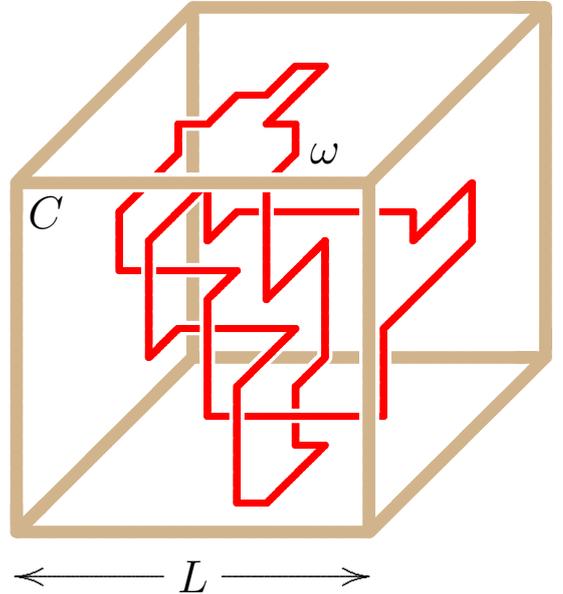

\begin{center}
\hrule\vspace{3mm}

\beginpicture
\setcoordinatesystem units <1.1pt,1.1pt>
\setplotarea x from -20 to 190, y from 0 to 190

\arrow <10pt> [.2,.67] from 70 -15 to 120 -15
\put {{\LARGE $L$}} at 60 -15 
\arrow <10pt> [.2,.67] from 50 -15 to 0 -15
\put {{\LARGE $\omega$}} at 105 130
\put {{\LARGE $C$}} at 10 110 

\color{Tan}
\setplotsymbol ({\scalebox{1.25}{$\bullet$}})
\plot 60 60 180 60 180 180 60 180 60 60 /
\plot 60 60 0 0 0 120 60 180 /

\color{white}
\setplotsymbol ({\scalebox{1.50}{$\bullet$}})
\plot 45 100 65 120 65 100 75 110 135 110 135 100 155 120 155 100
125 70 125 40 65 40 65 80 /
\color{red} 
\setplotsymbol ({\scalebox{0.75}{$\bullet$}}) 
\plot 45 100 65 120 65 100 75 110 135 110 135 100 155 120 155 100
125 70 125 40 65 40 65 80 /

\color{white}
\setplotsymbol ({\scalebox{1.50}{$\bullet$}})
\plot 105 100 105 60 95 50 95 30 105 30 85 10 75 10 75 50 95 70 55 70 45 60 45 100 /
\color{red} 
\setplotsymbol ({\scalebox{0.75}{$\bullet$}})
\plot 105 100 105 60 95 50 95 30 105 30 85 10 75 10 75 50 95 70 55 70 45 60 45 100 /

\color{white}
\setplotsymbol ({\scalebox{1.50}{$\bullet$}})
\plot 65 80 75 90 35 90 35 110  55 130 55 140 65 140 75 150 85 150 95 160 105 160
        95 150 85 140 95 140 95 130 85 120 85 80 105 100 /
\color{red} 
\setplotsymbol ({\scalebox{0.75}{$\bullet$}}) 
\plot 65 80 75 90 35 90 35 110  55 130 55 140 65 140 75 150 85 150 95 160 105 160
        95 150 85 140 95 140 95 130 85 120 85 80 105 100 /

\color{white}
\setplotsymbol ({\scalebox{1.50}{$\bullet$}})
\plot 45 60 45 100 65 120 /
\plot 65 40 65 80  /
\plot 90 40 100 40 /

\color{red} 
\setplotsymbol ({\scalebox{0.75}{$\bullet$}}) 
\plot 55 70 45 60 45 100 65 120 65 100  /
\plot 70 40 65 40 65 80 75 90 /
\plot 85 40 105 40 /
\plot 105 95 105 100 100 95 /

\color{white}
\setplotsymbol ({\scalebox{2.0}{$\bullet$}})
\plot 0 0 120 0 180 60 /
\plot 120 0 120 120 /
\plot 0 120 120 120 180 180 /
\color{Tan}
\setplotsymbol ({\scalebox{1.25}{$\bullet$}})
\plot 10 10 0 0 0 10 /  \plot 170 60 180 60 180 70 /
\plot 0 110 0 120 10 130 / \plot 180 170 180 180 170 180 /
\plot 0 0 120 0 180 60 /
\plot 120 0 120 120 /  \plot 0 120 120 120 180 180 /

\setlinear
\color{black}
\normalcolor
\endpicture

\end{center}
\caption{\textit{A lattice knot $\omega$ of knot type $5_1^+$ inside a cube in the
cubic lattice of side-length $L$ sites and volume $L^3$ lattice sites.  If $\omega$
has length $n$ (steps or lattice sites), then the concentration of the vertices 
(monomers) in the lattice knot is $\phi=n/L^3$.}
}
\label{3}  
\end{figure}

As anticipated in the introduction, the location of the knotted portion 
in a ring polymer confined within a cavity is not well-defined.  This is 
illustrated in figures~\ref{1} and~\ref{2}.  The panels in figure \ref{1} 
are schematic diagrams of a knotted ring polymer in a cubical cavity
at low concentration.  In the left panel the knot is confined to a small volume
in the cavity.   In this case the knot can be quarantined inside a small imaginary 
topological ball, forming a \textit{knotted ball-pair} \cite{deG84,SW88} with a
well-defined knot type.  The rest of the ring polymer outside the ball is also
knotted ball pair, and has knot type the unknot.  This construction gives
well-defined knot types to \textit{knotted arcs} in a ring polymer.  If the 
knotted ball pairs are, on average, statistically small, then the knot is said to be
``small'' or ``localized''.  This has been examined numerically in reference
\cite{OTJW98}, showing that the knotted ball-pairs or arcs in unconstrained lattice
polygons, sampled uniformly, are localized (see also reference 
\cite{RDKM08,RMSS13} for additional results).

The right panel in Figure~\ref{1} is a schematic of a knotted ring polymer
with a ``delocalized'' knot filling the entire cubical cavity.  In this case a 
small ball cannot be found to contain an entangled arc in the ring polymer 
with a non-trivial knot type.  While the knot is generally localized as 
in the left panel, delocalized knotted conformations may dominate 
if the ring polymer is more rigid, or when its persistence length is large, 
compared to the size of the confining cavity.

\begin{figure}[h!]
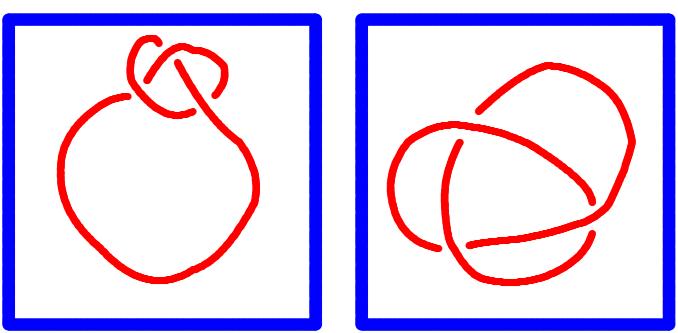

\begin{center}
\hrule\vspace{3mm}

\beginpicture
\setcoordinatesystem units <1.15pt,1.15pt>
\setplotarea x from 0 to 100, y from 0 to 100
\color{blue}
\setplotsymbol ({\scalebox{1.25}{$\bullet$}})
\plot 0 0 100 0 100 100 0 100 0 0 /
\color{red} \setquadratic
\setplotsymbol ({\scalebox{0.75}{$\bullet$}})
\plot 45 80 53 90 60 90 65 89 70 85 70 80 67 75 /
\plot 60 70 50 70 40 80 40 87 43 93 47 94 49 92 /
\plot 55 86 65 70 75 60 80 50 80 40 70 25 60 18 45 15 30 25
20 37 17 50 20 62 30 72 34 74 39 75   /

\setlinear
\color{black}
\normalcolor

\setcoordinatesystem units <1.15pt,1.15pt> point at -115 0 
\setplotarea x from -5 to 100, y from 0 to 100
\color{blue}
\setplotsymbol ({\scalebox{1.25}{$\bullet$}})
\plot 0 0 100 0 100 100 0 100 0 0 /
\color{red} \setquadratic
\setplotsymbol ({\scalebox{0.75}{$\bullet$}})
\plot 25 25 15 30 10 40 10 50 15 60 25 65 35 65 45 63 55 59 70 48 75 40  /
\plot 75 30 70 22 60 16 50 14 40 15 35 18 30 25 28 30 27 40 28 50 32 60 /
\plot 38 70 50 80 60 85 70 83 80 77 85 70 88 60 85 50 80 39 75 35 70 33
60 30 50 28 40 27 35 26 /
\setlinear
\color{black}
\normalcolor
\endpicture

\end{center}
\caption{\textit{Schematic diagrams of a knotted ring polymer in a cubical 
cavity at low concentration.  The conformation on the left shows a knot 
confined in a small volume inside the cavity, while the location of the 
knot on the right fills the volume of the cavity.}
}
\label{1}  
\end{figure}

\begin{figure}[h!]
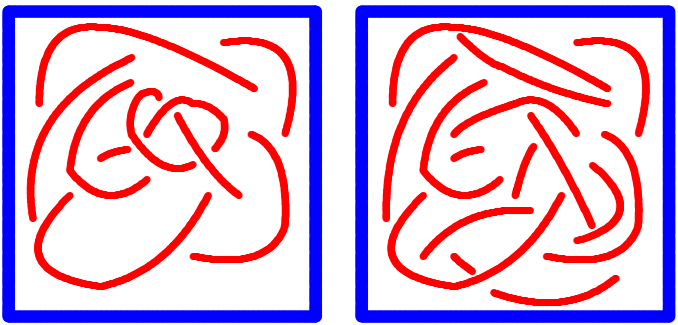

\begin{center}
\hrule\vspace{3mm}

\beginpicture
\setcoordinatesystem units <1.15pt,1.15pt>
\setplotarea x from 0 to 100, y from 0 to 100
\color{blue}
\setplotsymbol ({\scalebox{1.25}{$\bullet$}})
\plot 0 0 100 0 100 100 0 100 0 0 /
\color{red} \setquadratic
\setplotsymbol ({\scalebox{0.75}{$\bullet$}})
\plot 45 60 53 70 60 70 65 69 70 65 70 60 67 55 /
\plot 60 50 50 50 40 60 40 67 43 73 47 74 49 72 /
\plot 55 66 65 50 75 40 /  \plot 30 52 34 54 39 55 /
\plot 45 45 32 40 20 48 26 66 40 77 /
\plot 20 40 10 20 30 10 50 20 65 40 /
\plot 10 70 15 90 30 95 50 90 80 75 /
\plot 70 90 90 85 90 60 /
\plot 79 60 88 50 90 30 80 20 60 20 /
\plot 40 85 13 63 8 32 /

\setlinear
\color{black}
\normalcolor

\setcoordinatesystem units <1.15pt,1.15pt> point at -115 0 
\setplotarea x from -5 to 100, y from 0 to 100
\color{blue}
\setplotsymbol ({\scalebox{1.25}{$\bullet$}})
\plot 0 0 100 0 100 100 0 100 0 0 /
\color{red} \setquadratic
\setplotsymbol ({\scalebox{0.75}{$\bullet$}})
\plot 55 66 65 50 75 30 /  \plot 30 52 34 54 39 55 /
\plot 45 45 32 40 20 48 26 66 40 77 /
\plot 20 40 10 20 30 10 50 20 65 40 /
\plot 10 70 15 90 30 95 50 90 80 75 /
\plot 70 90 90 85 90 60 /
\plot 79 60 88 50 90 30 80 20 60 20 /
\plot 30 85 13 63 8 32 /
\plot 30 60 37 65 50 70 60 70 70 60 / 
\plot 20 20 35 32 55 35 /
\plot 30 20 32 18 36 15 / \plot 43 8 65 5 83 13 /
\plot 50 40 53 50 56 56 /
\plot 32 92 40 85 46 82 62 75 80 70  /
\plot 70 25 84 35 75 50 /

\setlinear
\color{black}
\normalcolor
\endpicture

\end{center}
\caption{\textit{Schematic diagrams of a knotted ring polymer in a cubical 
cavity at intermediate and high concentrations.  The knot is still visible in
the intermediate concentration on the left as a tight tangle towards
the centre, but if it does not remain tight with increasing concentration,
then it relaxes to fill the entire cavity in the high concentration 
regime on the right.}
}
\label{2}  
\end{figure}

The localization of a knotted arc in a confined ring polymer becomes
more intractable at medium and high concentrations when the
lattice polygon approaches a Hamiltonian state and fills the entire
cubical cavity.  This is illustrated schematically in Figure \ref{2}.
A knotted arc can still be seen in the left panel, but in the
high concentration phase, it is melted into a sea of tangles and 
cannot be easily located.  It is still an open question whether
the knotted arc becomes statistically delocalized with increasing 
concentration \cite{JOT25}.

\section{The free energy}

\subsection{Lattice polygon}

The number of lattice polygons of length $n$ (counted modulo translations),
denoted $p_n$, has long been studied as a model of ring polymer 
entropy \cite{H61A,deG79}.  It is known that
\begin{equation}
\lim_{n\to\infty} \frac{1}{n} \log p_n = \log \mu >0
\label{e1}
\end{equation}
exists, and this defines the \textit{growth constant} $\mu$ of lattice polygons.
The lattice polygon growth constant is equal to the self-avoiding walk 
growth constant \cite{HW62A}.  In the cubic lattice it has been calculated 
to high accuracy using sophisticated implementations of the pivot algorithm 
by Clisby \cite{C13} ($\mu=4.684039931(27)$).

Theoretical arguments \cite{LGZJ77,W81,GZJ98} and a tremendous amount of 
numerical data \cite{CS09} suggest that 
\begin{equation}
p_n \sim C \, n^{\alpha-3} \, \mu^n .
\label{e2}
\end{equation}
This shows dependence of $p_n$ on the critical exponent $\alpha$, which
in the underlying $O(N)$ model is the \textit{specific heat} exponent
\cite{W81}.  In three dimensions $\alpha = 0.237 \pm 0.002$ \cite{LMS95,GZJ98}.

The grand canonical partition function of lattice polygons is defined by
\begin{equation}
Z(x) = \sum_{n\geq 0} p_n\, x^n .
= \sum_{n\geq 0} p_n \, e^{-\upsilon n /k_BT},
\label{e3}
\end{equation}
where $T$ is the absolute temperature, $k_B$ 
the Boltzmann's constant and  $\upsilon$ the chemical potential.
The corresponding limiting grand canonical free energy 
(or \textit{grand potential})  is defined by
\begin{equation}
\Phi(x) = \log Z(x) .
\label{e4}
\end{equation}
By equations \Ref{e1} and \Ref{e2} one observes that
\begin{equation}
\Phi(x) \sim | x - x_c |^{2-\alpha},
\label{e5}
\end{equation}
where $x_c = 1/\mu$ is a critical point and radius of convergence of $Z(x)$ in
equation \Ref{e3}, separating phases dominated by finitely long, and 
by infinitely long, lattice polygons.

The decay of the $2$-point function of the self-avoiding walk is controlled
by its correlation length $\xi$ \cite{LMS95}, whose scaling is controlled 
by the \textit{metric exponent} $\nu$ in the finite phase, that is, when $x<x_c$:
\begin{equation}
\xi \sim |x - x_c|^{-\nu} .
\label{e6}
\end{equation}
All the metric quantities of polygons of length $n$, such as the root mean 
square radius of gyration $R_g$, or the mean span, of polygons of length 
$n$, have the large $n$ behaviour ruled by $\nu$, namely
\begin{equation}
R_g \sim n^\nu.
\label{e7}
\end{equation}  
In the cubic lattice $\nu = 0.587597 (7)$ \cite{C10} and if we assume 
that the hyperscaling relation $2-\alpha=d\nu$ holds, we get the 
independent estimate $\alpha=0.237209(21)$. The Flory value
of $\nu=3/5$ \cite{F69}, and this gives $\alpha=1/5$.

\subsection{Lattice knots}

The number of lattice knots of length $n$ and knot type $K$ is
denoted $p_n(K)$.  If $K$ is the unknot $0_1$, then it is known
that the limit
\begin{equation}
\lim_{n\to\infty} \frac{1}{n} \log p_n (0_1) = \log \mu_{0_1}  >0
\label{e8}
\end{equation}
exists \cite{SW88,JvRW90,WJvR93}. If $K$ is not the unknot (but is a non-trivial knot), then
\begin{equation}
\limsup_{n\to\infty} \frac{1}{n} \log p_n (K) = \log \mu_K > 0 . 
\label{e9}
\end{equation}
It is known that $\mu > \mu_K \geq \mu_{0_1} > 0$ \cite{SW88,P89}. 
Numerical data on these growth constants cannot, so far, rule out either 
$\mu_K = \mu_{0_1}$, or $\mu_K > \mu_{0_1}$, for arbitrary 
$K\not=0_1$.  For the purposes of this paper we conjecture that 
$\mu_K = \mu_{0_1}$.  Numerical results in reference \cite{JvRW91}
show that $|\mu - \mu_{0_1}| \approx 4.15\times 10^{-6}$, so that
$\mu_{0_1}$ and $\mu$ agree to (perhaps) $4$ decimal places.

The asymptotic scaling of $p_n$ in equation \Ref{e2} suggest that
\begin{equation}
p_n(K) \sim C_K \, n^{\alpha_K-3} \, \mu_K^n .
\label{e10}
\end{equation}
If $K$ is the unknot $0_1$, then numerical simulations suggest that
$\alpha_{0_1} = \alpha$ \cite{OTJW96,OTJW98,BO12}, in particular since 
the partition function $Z(x)$ in equation \Ref{e3}, for small values of $x$ 
are dominated by unknotted polygons \cite{JvRW91}.  However, this 
remains unproven.  

Numerical data in references \cite{OTJW96,OTJW98} suggest that
\begin{equation}
\alpha_K = \alpha_{0_1} + N_K 
\label{e11}
\end{equation}
where $N_K$ is the number of prime components in the knot type $K$
(see also references \cite{JvR09,JvR19a}).
These observations show that, asymptotically, in a lattice knot of knot 
type $K$, the prime components of $K$ are accommodated as local and
independent knotted arcs, each similar to the schematic in the left panel
in figure \ref{1}. 

Metric properties of lattice knots exhibit scaling similar to that of
lattice polygons in equation \Ref{e7}. Denote the metric exponent
associated with lattice knots of type $K$ by $\nu_K$. Then the root 
mean square radius of gyration of a lattice knot has a scaling
\begin{equation}
R_g (K) \sim n^{\nu_K} .
\label{e12}
\end{equation}
Numerical simulations support the conjecture that $\nu_K
= \nu_{0_1} = \nu$ \cite{JvRW91A}.

\subsection{Confined lattice knots}

An \textit{$L$-cube} in the cubic lattice is a cubical sublattice of $L^3$ 
lattice sites and side-length $L$ (lattice sites).  We place the origin of 
the cubic lattice at the \textit{bottom site} of the $L$-cube (this is the 
lexicograph least site in the cube).

Realisations of a lattice knot inside an $L$-cube has several degrees of
freedom, notably its length $n$, conformational degrees of freedom, 
and then also translational and rotational degrees of freedom.  Suppose
that there are $p_{n,L}(K)$ realisations of a lattice knot of fixed knot 
type $K$, length $n$, in an $L$-cube.  If these conformations are uniformaly
weighted, then the free energy of the lattice knot is $F_{n,L} (K) 
= - \log p_{n,L} (K)$ (this defines the canonical ensemble of this model).  
In the grand canonical ensemble we introduce a chemical potential 
$\upsilon>0$ and define the grand canonical partition function
\begin{equation}
Z_{K,L} (x) 
= \sum_{n\geq 0} p_{n,L}(K) \, x^n
\label{e13}
\end{equation}
where $x = e^{-\upsilon n /k_BT}$. We also use the convention that 
$p_{0,L}(K)=1$ for all $L$ and $K$, so that $Z_{K,L}(x) = 1 + \ldots$.  
Then the grand canonical free energy density of the model is given by
\begin{align}
f_{K,L}(x) &= \frac{1}{L^3} \log \sum_{n\geq 0} p_{n,L}(K) \, x^n .
\label{e14}
\end{align}
For finite values of $L$, $n \leq L^3$, so that $Z_{K,L}(x)$ is a 
polynomial in $x$ with non-negative coefficients and constant term
equal to $1$ (so that $f_{K,L}(0) = 0$).

\subsection{The limiting free energy}

The (canonical) free energy in equation \Ref{e14} has limit given by
\begin{align}
\chi_K(x)  &= \limsup_{L\to\infty} f_{K,L} (x) \cr 
&= \limsup_{L\to\infty} \frac{1}{L^3} \log Z_{K,L} (x) .
\label{e15}
\end{align}
If the limit exists then this defines a thermodynamic limit in the model.

If $K$ is the unknot ${0_1}$, then consider two lattice knots $\alpha$
and $\beta$, both of knot type the unknot.  We sample $\alpha$ uniformly
from an $L$-cube, and $\beta$ similarly from an $M$-cube.

A \textit{bottom edge}, and a \textit{top edge}, of $\alpha$
can be defined by a lexicographic ordering of the edges along $\alpha$
by the coordinates of their midpoints.  Similarly, the bottom and top
edges of $\beta$ can be located.  The polygons $\alpha$ and $\beta$
can be \textit{concatenated} into a single polygon by translating
and rotating $\beta$ such that the midpoint of its bottom edge is
one step in the $x$-direction \textit{above} the midpoint of the
top edge of $\alpha$.  These bottom and top edges are then the opposing
sides of a unit square in the lattice, and by deleting them, and then
adding back in the \textit{other} two edges of the unit square, 
a single lattice knot, $\alpha\circ\beta$, composed of the sites 
along $\alpha$ and $\beta$, is obtained \cite{H61A,HW62A,JvRW90,WJvR93,JvR02}.

If both $\alpha$ and $\beta$ are lattice realisations of the unknots, then
the knot type of $\alpha\circ\beta$ is the unknot.  If $\alpha$ has 
length $k$, then there are at most $p_{k,L}({0_1})$ choices for 
it, and if $\beta$ has length $n-k$, then it can be chosen in at most 
$p_{n-k,L}({0_1})/2$ ways (the division by $2$ accounts for the
rotation of $\beta$ to align the top and bottom edges).  

The concatenated lattice knot $\alpha\circ\beta$ has length $n$, and 
can be placed inside an $(L+M)$-cube in at most $(L+M)^3$ ways. This gives the
inequality
\begin{equation}
\sum_{k=0}^n p_{k,L}({0_1}) \, p_{n-k,M}({0_1})
\leq 2(L+M)^3 p_{n,L+M}({0_1}) .
\label{e16}
\end{equation}

This inequality gives the following theorem.

\begin{theorem}
The limit $\displaystyle \chi_{0_1} (x) = 
\lim_{L\to\infty} \frac{1}{L^3} \log Z_{{0_1},L}(x)$ exists and
defines the limiting free energy of confined lattice unknots.  Moreover,
since $Z_{{0_1},L}(x) \geq 1$, it follows that
$\chi_{0_1} (x) \geq 0$.
\label{thm1}
\end{theorem}

\proof 
Multiply equation \Ref{e7} by $x^n$ and sum over $n$.  This gives
\[ Z_{{0_1},L}(x) \, Z_{{0_1},M}(x) \leq
2(L+M)^3 \, Z_{{0_1},L+M} (x) . \]
That is, $\log Z_{{0_1},L}(x)$ satisfies a generalised superadditive
relation of the kind examined by Hammersley \cite{H62}, proving the
existence of the limit as claimed. 
\qed

\subsection{Scaling of the free energy}

The function $p_{n,L}({0_1})$ increases exponentially with $n$
when $n\ll L$.  That is, according to equation~\Ref{e8}, 
$p_{n,L}({0_1}) = \mu_{0_1}^{n+o(n)}$ if $L$ is large
and $n\ll L$.  This suggests that, if $K={0_1}$, 
the summation in equation \Ref{e13} is convergent in the limit $L\to\infty$
if $x\mu_{0_1} < 1$, while it is divergent if $x\mu_{0_1} > 1$.
When $x\mu_{0_1} < 1$, then one expects $\chi_{0_1} (x) = 0$,
while $\chi_{0_1} (x) > 0$ if $x\mu_{0_1} > 1$.

Thus, for the unknot ${0_1}$, the limiting free energy
$\chi_{0_1} (x)$ has a \textit{critical point} $x_{0_1}
=1/\mu_{0_1}$ (theorem \ref{thm1}).  In other words,
$\chi_{0_1} (x)$ is a non-analytic function, and the standard 
assumption is that its \textit{singular part} has \textit{scaling} behaviour
given by
\begin{equation}
\chi_{0_1} (x) \sim
\begin{cases}
0, & \hbox{if $x< x_{0_1}$}; \\
|x {-} x_{0_1}|^{2-\alpha_{0_1}^{*}} + \hbox{[...]}, 
& \hbox{if $x>x_{0_1}$},
\end{cases}
\label{e17}
\end{equation}
where $\hbox{[...]}$ are terms dominated by the leading singular term, and
$\alpha_{0_1}^{*}$ is the \textit{specific heat} exponent associated 
with the non-analyticity in $\chi_{0_1} (x)$.  The exponent 
$\alpha_{0_1}^{*}$ in equation \Ref{e17} is in particular related
to the (collapse) phase transition of \textit{confined lattice unknots}
and it should not be confused with the entropic exponent $\alpha_{0_1}$
introduced in equations \Ref{e10} and \Ref{e11} that refers  unknotted polygons in free space.  
Nor should it be
confused with the entropic exponent of topologically unconstrained lattice polygons in equation~\Ref{e2}.  While it may be the case that $\alpha_{0_1} = \alpha$, the confinement of lattice knots changes the entropic properties of
the lattice knots and may change the value of the associated set of
critical exponents related to the collapse transition of confined lattice unknots, including the value of $\alpha_{0_1}^*$.

Although theoretically less well understood, similar scaling behaviour is  expected for non-trivial knot types $K$, namely a critical exponent
$\alpha_K^*$ controlling the corresponding non-analytic behaviour of the free energy $\chi_K(x)$ of a confined lattice knot of type $K$.  In this event one expects, similar to equation \Ref{e17}, a critical point $x_K$
in the free energy $\chi_K(x)$, such that
\begin{equation}
\chi_K (x) \sim
\begin{cases}
0, & \hbox{if $x< x_K$}; \\
|x {-} x_K|^{2-\alpha_K^{*}} + \hbox{[...]}. 
& \hbox{if $x>x_K$},
\end{cases}
\label{e18}
\end{equation}
This poses the following additional questions: (1) Is it the case that $x_K=x_{0_1}$?
It can be shown that $x_{0_1} \leq x_K$ \cite{JvR09}.
(2) What is the relationship between the \textit{thermodynamic exponents} 
$\alpha_{0_1}^{*}$ and $\alpha_K^{*}$?

In the scaling limit, the \textit{energy density} $\mathcal{E}_K$ 
(or the density of occupied sites) of the confined lattice knot is 
the derivative of $\chi_K(x)$ to $\log x$.  This gives the scaling
\begin{eqnarray}
\label{e19}
\mathcal{E}_K (x) &=& x\frac{d}{dx}\, \chi_K(x)  \\
&\sim&
\begin{cases}
0, & \hbox{if $x< x_K$}; \nonumber \cr
|x {-} x_K|^{1-\alpha_K^{*}} + \hbox{[...]},
& \hbox{if $x>x_K$}.
\end{cases} 
\end{eqnarray}
Notice that the (limiting) density (or concentration) of polymer in the cube 
is given by $\mathcal{E}_K (x)$.

The second derivatice of $\chi_K(x)$ is the variance or \textit{specific heat} of
the model.  In this case we have
\begin{eqnarray}
\label{e20}
\mathcal{C}_K (x) &=& x\frac{d}{dx}\, \mathcal{E}_K(x) \\
&\sim&
\begin{cases}
0, & \hbox{if $x< x_K$};  \nonumber \cr
|x {-} x_K|^{-\alpha_K^{*}} + \hbox{[...]},
& \hbox{if $x>x_K$}.
\end{cases} 
\end{eqnarray}

\section{Numerical sampling of confined lattice knots}

Confining the lattice knot to an $L$-cube reduces the state space of the
model, resulting in smaller irreducibility classes of conformations.
The number of lattice knots of knot type $K$, length $n$, inside an 
$L$-cube (and so in an irreducibility class), is denoted $p_{n,L}(K)$.
Confined lattice knots in $L$-cubes were sampled as in references
\cite{JvRR11A,JvR19a,JOT25}.  The GAS algorithm \cite{JvRR09,JvRR11A},
implemented with BFACF elementary moves \cite{BF81,ACF83,JvRW91},
was used to sample along tours (or Markov Chains) in the state space 
of lattice knots in an $L$-cube. The algorithm is an \textit{approximate 
enumeration algorithm} (see, for example, reference \cite{JvR09A}).  
Its implementation returns ratio estimates $p_{n,L}(K)/p_{m,L}(K)$ for 
confined lattice knots of lengths $n$ and $m$.  For details of the
implementation to sample lattice knots, see references \cite{JvRR11A,JOT25}.

In addition to newly generated data using the GAS algorithm, we used 
data collected in reference \cite{JOT25}, and where necessary, supplemented 
it with additional simulations using the GARM algorithm \cite{RJvR08,JvRR09} 
to check for consistency and convergence.  The GAS algorithm was 
implemented in parallel to sample aloung sequences, one per CPU, using
omp protocols. Typically, tours were generated along $T$ parallel tours, 
collected in $M$ blocks, each block of length $10^7$ iterations BFACF 
elementary moves per tour.  For example, for the trefoil knot type in a 
cube of side length $15$ a total of $M=343$ blocks were sampled, each 
block consisting of $4$ parallel tours, and each tour of length $10^7$ iterations.  
This gives a total of $2.576 \times 10^{10}$ iterations.  Since the blocks 
are independently sampled, confidence intervals can be calculated by 
treating estimates from each block as independent.

Confined lattice knots were sampled in $L$-cubes for $L\in\{7,9,11,13,15\}$.
In an $L$-cube the maximum length polygon has length $L^3-1$ (if $L$ is
odd).  Thus, in the case of $L=13$, for example, the algorithm sampled
confined lattice knots of lengths up to $n=2196$. Sampling deteriorated
as the maximum length was approached.

The minimal length of a realisation of the unknot in the cubic lattice is 
$4$ (this is the boundary of a plaquette).  Thus, it follows that in the cubic
lattice, $p_4(\emptyset) = 3$.  The minimal length lattice unknot can also
be placed in $3L(L{-}1)^2$ distinct ways in an $L$-cube. That is,
$p_{4,L}(\emptyset) = 3L(L{-}1)^2$.   Since the GAS algorithm estimates
ratios $p_{n,L}(\emptyset)/p_{m,L}(\emptyset)$, one can choose $m=4$ to
obtain estimates of the counts $p_{n,L}(\emptyset)$ of lattice unknots of
length $n$ in an $L$-cube.

The situation is similar for non-trivial knot types.  For example, the
minimum length of a lattice knot of knot type the trefoil is $24$
\cite{D93}, and there are $1664$ distinct placements of the
minimal length right-handed trefoil in the cubic lattice lattice (equivalent under
translations) \cite{ISDAVS12}. One may also explicitly, by computer, count 
the number of ways a minimal length trefoil can be placed in an $L$-cube
\cite{JOT25}.  For example, in a $3$-cube there are $p_{24,3} (3_1^{+})=2084$
placements of the minimal length right-handed lattice trefoil. This can be 
used to normalise the data sampled by the GAS algorithm by 
choosing $m=24$ in the estimates of the ratios
$p_{n,L}(\emptyset)/p_{4,L}(\emptyset)$ when sampling right-handed
lattice trefoils.  For more details, see reference \cite{JOT25}.

Since data were collected in independent blocks, we were able to calculate
standard deviations on our estimates of $p_{n,L}(K)$.  These can be used
to determine statistical error bounds on our estimates of the free energy.  
For example, in the case of the unknot $f_{K,L}(1.0) = 0.7720(14)$ if $L=7$ 
$f_{K,L}(1.0) = 0.85372(19)$ for $L=15$.  Error bounds similar to these 
were determined for other knot types.

\subsection{Finite size scaling}

The scaling relations in equations \Ref{e18}--\Ref{e20} give the scaling
of the free energy and its derivatives in the limit as $L\to\infty$.  Generally
the free energy $\chi_K(x)$ is not known, and can only be approximated
by $f_{K,L}(x)$ for finite values of $L$.  The finite size free energy
$f_{K,L}(x)$ is similarly unknown, but can be approximated by numerical means. 

Taking the derivative of equation \Ref{e14} the finite size energy 
density is obtained:
\begin{eqnarray}
\mathcal{E}_{K,L}(x) 
&=& x\frac{d}{dx} f_{K,L}(x) \label{e21}  \\
&=& \frac{1}{Z_{K,L}(x)} \sum_{n\geq 0} (n/L^3)\,p_{n,L}(K)\,x^n. \nonumber
\end{eqnarray}
Since $n/L^3$ is the \textit{concentration} or density of occupied sites 
in the $L$-cube, $\phi_{K,L}(x)=\mathcal{E}_{K,L}(x)$ 
is the \textit{mean concentration} of 
the lattice knot, giving explicit meaning to the energy density in 
the finite size model.  

In general the free energy $f_{K,L}(x)$ is small for $x<x_K$, and 
it increases with $x$ when $x>x_K$ where the partition
function is dominated by long polygons.  This gives rise to two regimes
in the finite size model, namely a \textit{solvent} or \textit{empty}
phase  when $x<x_K$, and a \textit{polymer} or \textit{dense} phase
when $x>x_K$.  The crossover from the solvent to the polymer phase
occurs in the vicinity of $x_K$, which is the critical point separating the 
two phases in the scaling limit.

In the finite-size system thermodynamic quantities are functions of 
ratios of length scales in the model.  At any given value of $x$, there
are two length scales, namely the side length $L$ of the cube, and
then the correlation length $\xi$ defined in equation \Ref{e6}.
Following the arguments in reference \cite{BO22} one can assume that 
\begin{equation}
\mathcal{E}_{K,L}(x) \sim L^q \, h_0( (x{-}x_K)\,L^{1/\nu}),
\label{e22}
\end{equation}
where $h_0$ is a finite-size scaling function.
This introduces the exponent $q$, which is related to the concentration 
(of monomers or vertices) in the confining cube.  This relation can also be
written as 
\begin{equation}
\mathcal{E}_{K,L}(x) \sim (x{-}x_K)^{-q\nu}  h_1( (x{-}x_K)\,L^{1/\nu}) ,
\label{e23}
\end{equation}
where $h_1(z) = z^{q\nu} h_0(z)$.  The exponent $1/\nu$ controls 
the \textit{crossover} scaling in the $x$--$L$ diagram. Comparison with 
equation \Ref{e19} shows that, if $z<0$, then $h_0(z) \to 0$ as 
$L \to\infty$.  On the other hand, if $z>0$ and small, then 
$h_0(z) \sim  z^{1-\alpha_K^*}$ so that
\begin{equation}
q=(\alpha^*_K-1)/\nu .
\label{e24}
\end{equation}  

\begin{figure}[t!]	
\includegraphics[width=0.5\textwidth]{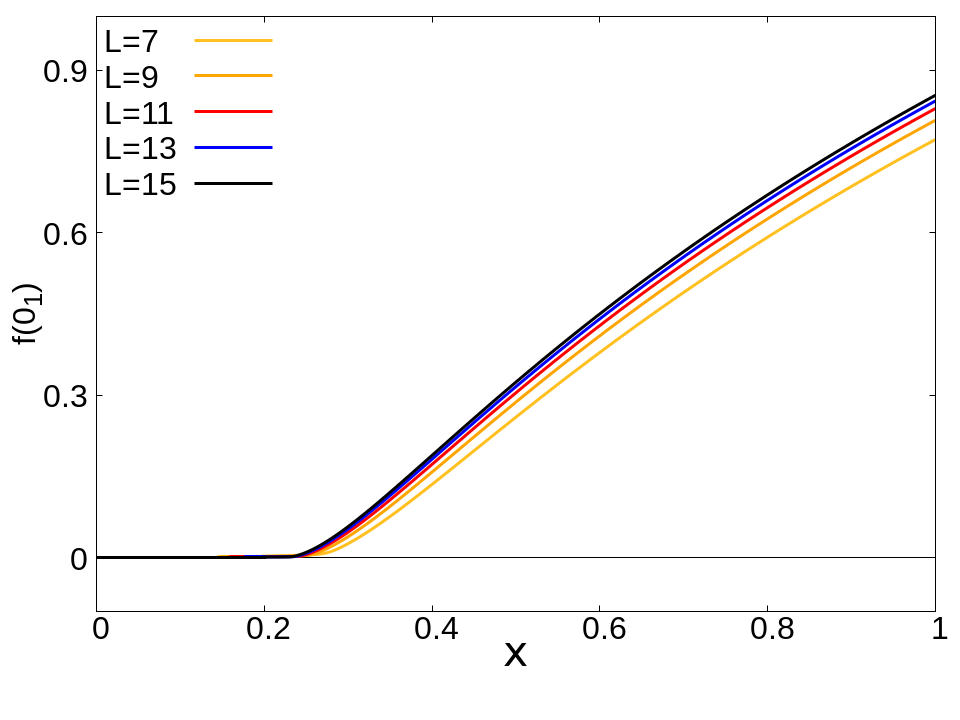}
	\caption{\textit{Finite size free energies $f_{0_1,L}(x)$ (see equation
	\Ref{e14}) of the unknot for $x\in[0,1]$. The values of $L$ increases
	from $L=7$ in steps of $2$ to $L=15$.  The curves accumulate on
	zero when $x$ is small, but diverges once $x$ is larger than a
	a critical point $x_{0_1}$, consistent with equation \Ref{e18}.}
	} 
\label{fFE01}
\end{figure}

\begin{figure}[h!]	
\includegraphics[width=0.5\textwidth]{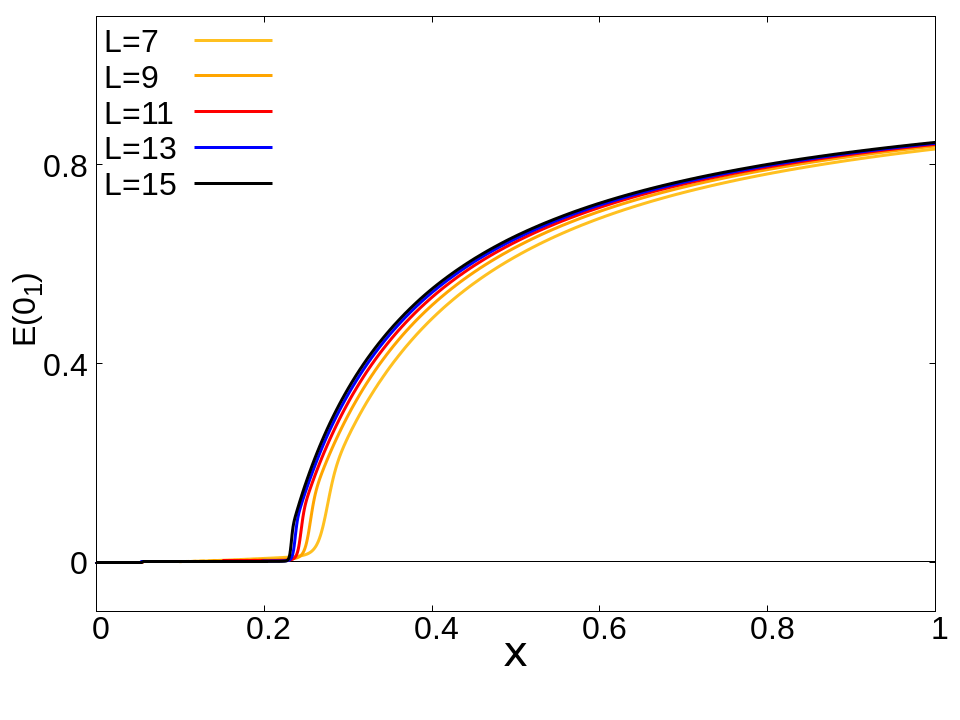}
	\caption{\textit{Finite size energy densities (see equation
	\Ref{e19}) as a function of $x$ for the unknot.  The curves
	accumulate on zero for $x$ small, but diverge sharply when 
	$x$ increases beyond the critical point. With increasing $x$
	$\C{E}_{0_1}$ approaches $1$, for example, if $x=2$,
	then $\C{E}_{0_1} \approx 0.93$ if $L=15$. As noted above
	equation \Ref{e19}, the density of polymer in the confining box 
	is given by $\C{E}_{K,L}(x)$.}
	} 
\label{fE01}
\end{figure}


\subsubsection{The lattice unknot $0_1$}

In figure \ref{fFE01} we plot the free energies $f_{0_1,L}(x)$ (see equation
\Ref{e14}) for confined unknotted polygons as a function $x$. The curves 
accumulate on zero for small values of $x$, but beyond a critical point the 
curves are increasing and appears to accumulate towards a limiting curve 
with increasing $L$.

The corresponding energy density curves $\mathcal{E}_{0_1,L}(x)$ are 
plotted in figure \ref{fE01} (see equation \Ref{e21}).  The sharp change 
at a finite size critical point $x_{0_1} (L)$ in each curve is a sign of a 
critical point $x_{0_1}$ in the thermodynamic limit (when $L=\infty$). 

\begin{figure}[h!]	
\includegraphics[width=0.5\textwidth]{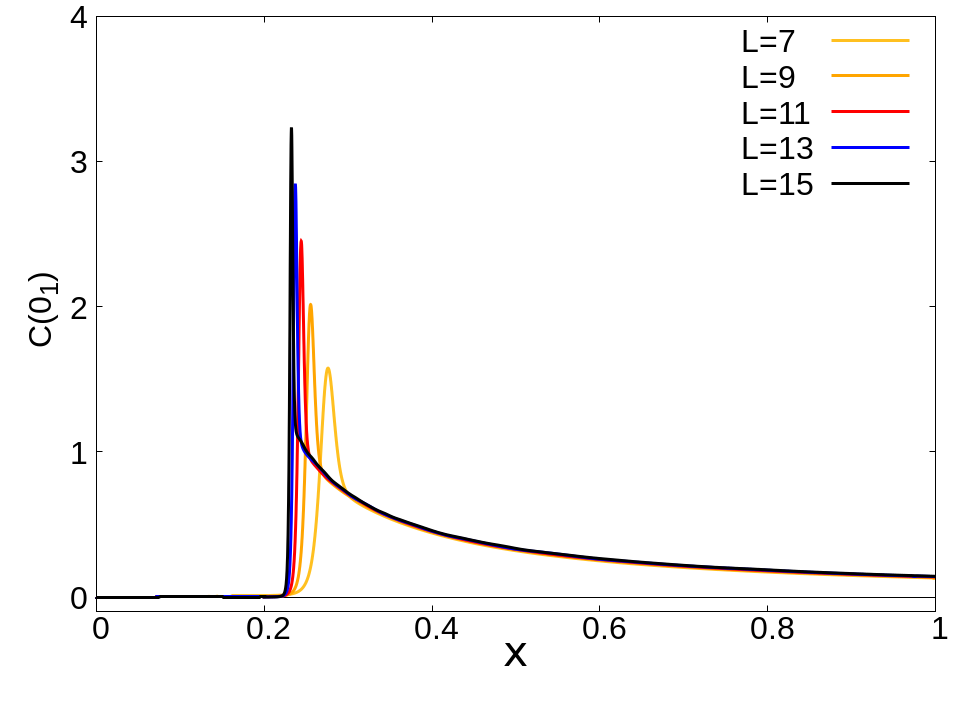}
	\caption{\textit{The finite-size specific heat curves for the unknot 
	(equation \Ref{e20}) exhibit a sharp peak near the critical point 
	for each chain length L. As L increases, these peaks become 
	progressively more pronounced — growing both taller and narrower
	— while the curves themselves appear to converge toward a limiting 
	profile. Notably, on the high-temperature side of each peak, 
	the curves develop an increasingly sharp transition as they 
	approach the maximum.}
	} 
\label{fC01}
\end{figure}

The specific heat curves are calculated from the second derivatives of the
free energy (equation \Ref{e20}), and these are plotted in figure 
\ref{fC01}.  With increasing $L$ these curves show a very sharp spike in 
the vicinity of the critical point.  The location of the spike at a point 
$x_{0_1}(L)$ is a finite size approximation of the thermodynamic critical 
point $x_{0_1}$.  Numerical estimates of $x_{0_1}(L)$ are obtained by locating
the position of the maximum in the specific heat curves (this is also the location
of the maximum height of the spike).  One-half of the width of the peaks at 
half-height were taken as an error bar in each case. The heights of the peaks 
or spikes in the specific heat curves, denoted $H_{0_1}(L)$, were also determined 
(this is also the maximum height of the specific heat). These data are 
listed in table \ref{t1}.

\begin{table}[h!]
\caption{Estimates of $x_{0_1}(L)$ and $H_{0_1}(L)$ for unknots}
\def\arraystretch{1.33}
\begin{ruledtabular}
\begin{tabular}{c  @{}*{1}{c} @{}*{1}{c}  }         
$L$  & $x_{0_1}(L)$ & $H_{0_1}(L)$ \cr   
\hline
  $7$ & $0.276 \pm 0.015$ & $1.5763$  \cr
  $9$ & $0.2553 \pm 0.0079$ & $2.0141$  \cr
  $11$ & $0.2441 \pm 0.0050$ & $2.4509$ \cr
  $13$ & $0.2371 \pm 0.0034$ & $2.8452$ \cr
  $15$ & $0.2325 \pm 0.0026$ & $3.2331$ \cr
\end{tabular}
\end{ruledtabular}
\label{t1}   
\end{table}

As noted below equation \Ref{e12} there are numerical evidence that
the metric exponent $\nu$ for a knotted polygon is equal to that of
the self-avoiding walk metric exponent.  Thus, we assume that, in equations
\Ref{e22} and \Ref{e23}, the crossover scaling is controlled by $1/\nu$,
where $\nu=0.587597(7)$ \cite{C10}.  Thus, assuming that $x_{0_1}(L)$ can
be extrapolated using a model
\begin{equation}
 x_{0_1}(L) = x_{0_1} + a/L^{1/\nu}, 
\label{e25}
\end{equation}
one may perform a linear fit to extract an estimate of the thermodynamic critical 
point $x_{0_1}$ in the model.  Since the estimates are independent, we proceed 
by bootstrapping the fits (by either adding or subtracting the error bar from 
each data point), and generating a sequence of $1000$ estimates.  The mean and 
variance of these estimates gave our final estimate of the critical point:
\begin{equation}
x_{0_1} = 0.2161 \pm 0.0069 .
\label{e26}
\end{equation}
Notice that the growth constant for all polygons (see equation \Ref{e1})
has best estimate $\mu = 4.684039931(27)$ \cite{C13} so that
\[ 1/\mu = 0.213490921(13) . \]  
This value is not excluded by our result, although it is known that 
$x_{0_1} \not= 1/\mu$ \cite{SW88,JvRW91}.  The 
best value for lattice polygons is \cite{C13} is
\[ 1/\mu = 0.2134909212 \pm 0.0000000013 . \]

Taking a derivative to $x$ of equation \Ref{e22} gives the finite size
scaling for the specific heat of a confined lattice knot of knot type $K$
\begin{align}
\mathcal{C}_{K,L}(x) \simeq 
\begin{cases}
& L^{q+1/\nu} \, h_K^\prime( (x{-}x_{K})\, L^{1/\nu}) ; \\
& (x-x_{K})^{1-q\nu}\, h_K^{\prime\prime}(  (x{-}x_{K})\, L^{1/\nu}) .
\end{cases}
\label{scalingK}
\end{align}
Generally, one expects a critical point $x_K$ (which may be function of
the knot type $K$), and scaling functions $h_K^\prime (z)$ and 
$h_K^{\prime\prime} (z)$.  In the event that $K=0_1$ the critical 
point is at $x_{0_1}$ and the scaling function is denoted by
$h_{0_1}(x)$.  This gives 
\begin{align}
\mathcal{C}_{0_1,L}(x) \simeq
L^{q+1/\nu} \, h_{0_1}^\prime( (x{-}x_{0_1})\, L^{1/\nu})  .
\label{e27}
\end{align}
It was argued after equation \Ref{e23} that $q\nu=\alpha_{0_1}^*-1$,
but the relationship with the entropic exponent of lattice polygons 
$\alpha$ (equations \Ref{e2} and \Ref{e5}) is unclear.  

The peak heights in the specific heat (figure \ref{fC01} and table \ref{t1}) 
are, by equation \Ref{e27}, proportional to $L^{q+1/\nu}$.  By taking ratios
of $H_{0_1}(L)$ in table \ref{t1}, one observes that $H_{L_1}/H_{L_2}
= (L_1/L_2)^{q+1/\nu}$.  Pairwise there are $10$ 
different ratios that can be constructed from the $H_{0_1}(L)$ in 
table \ref{t1}.  The minimum estimate obtained for $q$ is $-0.8089$, and 
the maximum estimate is $-0.7253$.  Determining the mean and the variance
of the estimates give the result
\begin{equation}
q = -0.765 \pm 0.035 , \qquad\hbox{(for the unknot $0_1$)}.
\label{e28}
\end{equation}

These estimates of the critical point $x_{0_1}$ (equation \Ref{e26}) and
$q$ (equation \Ref{e28}) makes possible the rescaling of the specific heat
curves by plotting $C_{0_1,L}(x)/L^{q+1/\nu}$ as a function
of $(x-x_{0_1})L^{1/\nu}$.  This is a test of the finite size scaling 
hypothesis in this model and this is displayed in figure \ref{fC01scaled}.

\begin{figure}[h!]	
\includegraphics[width=0.5\textwidth]{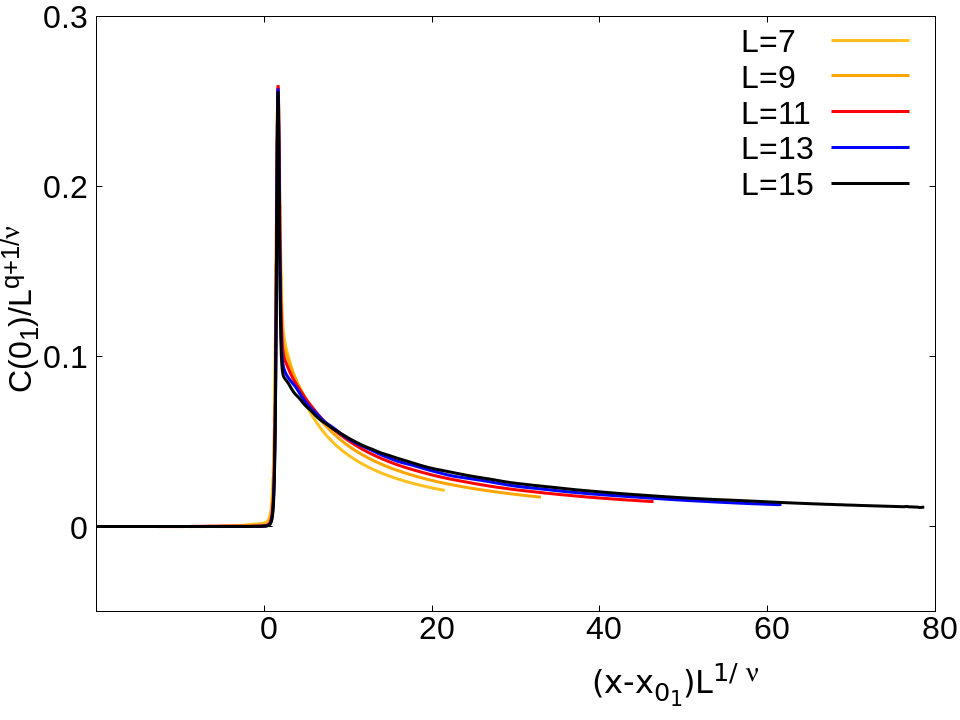}
	\caption{\textit{Rescaling the data in figure \ref{fC01} using the
	scaling hypothesis in equation \Ref{e27}.  The rescaling gives
	a narrower spike at the critical point, and as in figure \ref{fC01},
	the approach of the curve to the spike from above appears to 
	create a sharp turn at the critical point.}
	} 
\label{fC01scaled}
\end{figure}

In figure \ref{fC01scaled} the rescaled specific heat, according to equation
\Ref{e27}, is plotted. There appears to be a ``kink'' developing in 
this figure as the rescaled curves $\mathcal{C}_{0_1,L}/L^{q+1/\nu}$ 
curves approach $x_{0_1}$ from above (this is roughly at the point $(0,0.1)$
where the free energy curve joins the developing spike). This is already suggested by 
the specific heat data in figure \ref{fC01}.  The location of this \textit{junction}
rescales stably under vertical rescaling  by $L^q$ in figure \ref{fC01scaled} -- 
for different values of $L$, these data appear to converge to the same shape 
and location.  Thus, like the other features in figure \ref{fC01scaled} (for example, 
the height of the peak in the specific heat) it is subject to the same rescaling 
in the phase diagram.  

A schematic illustration of the rescaled data is shown 
in figure  \ref{Cscaled}, with two bullets, one labeled $t$ at the junction, and 
the other label $p$ at the peak of the specific heat curve.  These two points 
both are features in the data with rescaling controlled by the exponent $q$. 

\begin{figure}[h!]
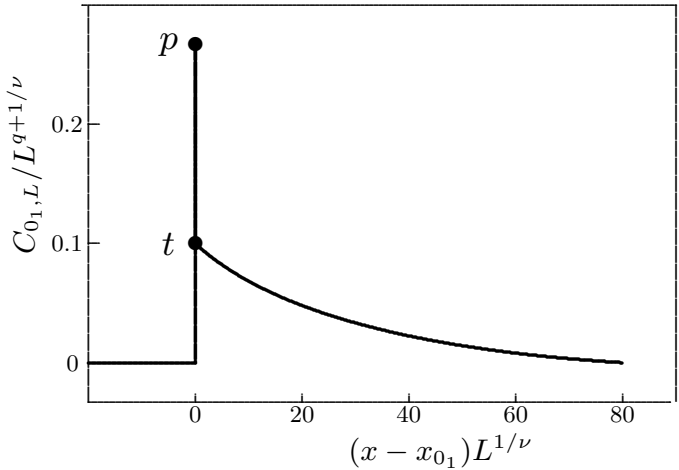

\beginpicture
\setcoordinatesystem units <2pt,1.5pt>
\setplotarea x from -10 to 100, y from -15 to 100

\put {\rotatebox{90}{\scalebox{1.25}{$C_{0_1,L}/L^{q+1/\nu}$}}} at -12 60 
\put {\scalebox{1.25}{$(x-x_{0_1})L^{1/\nu}$}} at 66 -12 

\plot 0 0 110 0 110 100 0 100 0 0 /

\multiput {\beginpicture \plot 0 0 0 2 / \endpicture} at 20 0 40 0 60 0 80 0 100 0  /
\put {$0$} at 20 -3 
\put {$20$} at 40 -3 
\put {$40$} at 60 -3 
\put {$60$} at 80 -3 
\put {$80$} at 100 -3

\multiput {\beginpicture \plot 0 0  2 0 / \endpicture} at 0 10 0 40 0 70  /
\put {$0$} at -3 10 
\put {$0.1$} at -4 40 
\put {$0.2$} at -4 70 

\setplotsymbol ({\scalebox{0.33}{$\bullet$}})
\plot 0 10 20 10 20 90 20 40 /
\setquadratic
\plot 20 40  50 20  100 10 /

\multiput {\scalebox{1.5}{$\bullet$}} at 20 40 20 90 /

\put {\scalebox{1.5}{$p$}} at 15 90
\put {\scalebox{1.5}{$t$}} at 15 40
\endpicture

\caption{\textit{A schematic diagram of the rescaled specific heat curves
denoting the general appearance of the curves in figure \ref{fC01scaled}.}
	} 
\label{Cscaled}
\end{figure}

Since $q=(\alpha_K^* - 1)/\nu$, the estimate for $q$ in equation \Ref{e28}
can be used to estimate the specific heat exponent for confined
unknotted polygons (see equations \Ref{e18}, \Ref{e19} and \Ref{e20}).
This gives
\begin{equation}
\alpha_{0_1}^* = 0.551 \pm 0.021 .
\label{e29}
\end{equation}

Finally, it was observed after equation \Ref{e10} that the entropic exponent
of lattice unknots $\alpha_{0_1} \approx \alpha \approx 0.237$ \cite{OTJW98,BO12}.  
Comparison with the estimate of $\alpha_{0_1}^*$ in equation \Ref{e29} shows 
that, numerically, $\alpha_{0_1}^*\not=\alpha_{0_1}$.  That is, the specific 
heat exponent $\alpha_{0_1}^*$ of \textit{confined lattice unknots} as it moves through
a \textit{collapse transition} from a solvent rich phase at small values of $x$ to
a polymer rich phase when $x>x_{0_1}$, is not equal to the entropic exponent 
$\alpha_{0_1}$ of lattice unknots. 

\begin{figure}[h!]	
\includegraphics[width=0.5\textwidth]{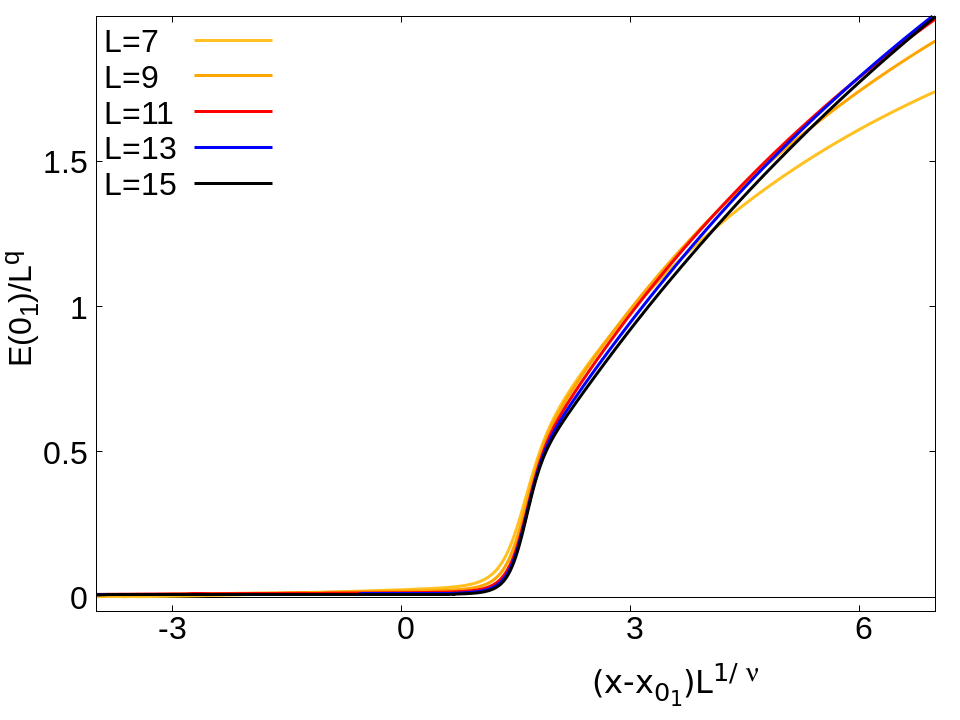}
	\caption{\textit{The rescaled energy density curves of the unknot.  The 
	curves coincide close to the critical point, forming a scaling region 
	for the model.}
	} 
\label{fE01scaled}
\end{figure}

The rescaled energy density curves, near the critical point $x_{0_1}$,
are finally plotted in figure \ref{fE01scaled}.  The curves for different values
coincide, exposing the scaling function $h_0$ (see equation \Ref{e22}).
By equation \Ref{e25}, $(x_{0_1}(L) - x_{0_1})\,L^{1/v} \approx a$,
where $a$ is a constant, and where $x_{0_1}(L)$ are finite size estimates
of the critical point $x_{0_1}$.   This is shown to good numerical accuracy 
in figure \ref{fE01scaled}.

\medskip

\subsubsection{The right-handed lattice trefoil $3_1^{+}$}

Numerical estimates of the finite size free energies $f_{3_1^{+},L}(x)$ 
(see equation \Ref{e14}) of the knot type $3_1^{+}$ are plotted in figure
\ref{fFE31}. These curves are similar to the data for unknots
in figure \ref{fFE01}, although the curves are somewhat more widely spaced
with changes in $L$.  The finite size energy densities for lattice trefoils, 
$\mathcal{E}_{3_1^{+},L}(x)$, are plotted in figure \ref{fE31}.  Comparison
to figure \ref{fE01} shows that the trefoil energy density deviates
from the corresponding curves for the unknot.

\begin{figure}[h!]	
\includegraphics[width=0.5\textwidth]{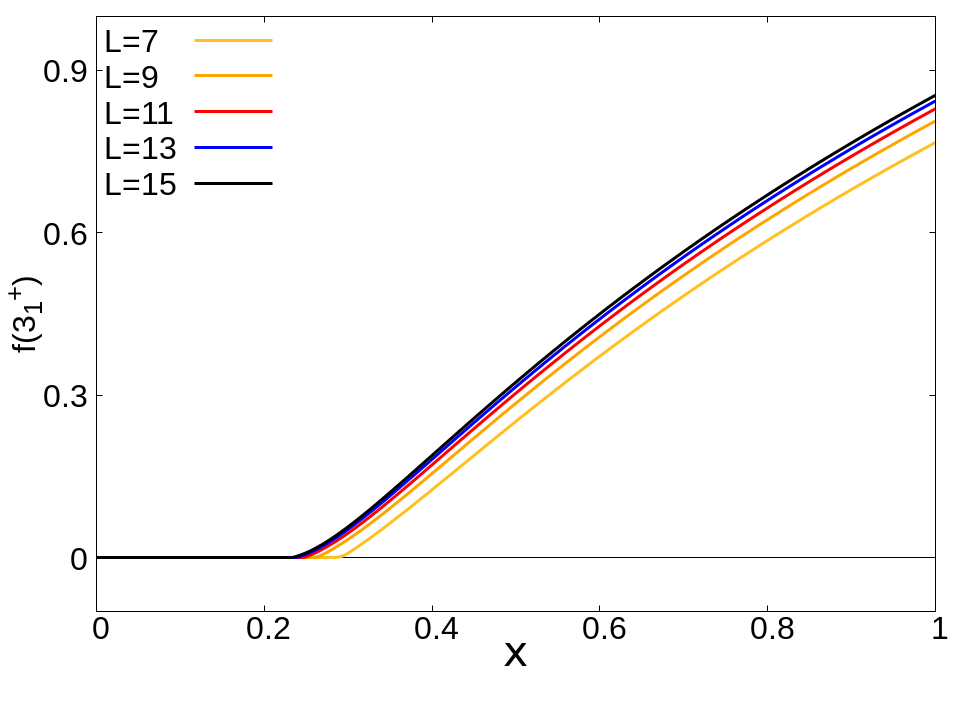}
	\caption{\textit{Finite size free energies $f_{3_1^{+},L}(x)$ (see equation
	\Ref{e14}) of the right handed trefoil for $x\in[0,1]$. The values of $L$ 
	increases from $L=7$ in steps of $2$ to $L=15$.  The curves 
	accumulate on zero when $x$ is small, but diverges once $x$ 
	is larger than a critical point $x_{3_1^{+}}$, consistent with 
	equation \Ref{e18}.}
	} 
\label{fFE31}
\end{figure}

\begin{figure}[h!]	
\includegraphics[width=0.5\textwidth]{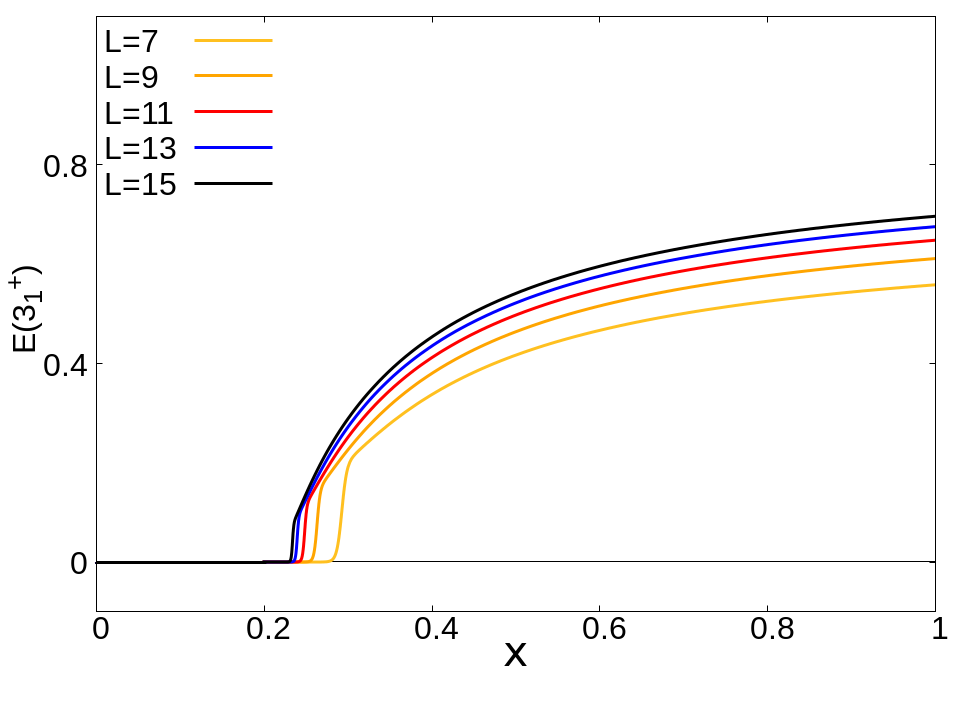}
	\caption{\textit{Finite size energy densities (see equation
	\Ref{e19}) as a function of $x$ for the right handed trefoil.  The curves
	accumulate on zero for $x$ small, but diverge sharply when 
	$x$ increases beyong the critical point.}
	} 
\label{fE31}
\end{figure}

\begin{figure}[h!]	
\includegraphics[width=0.5\textwidth]{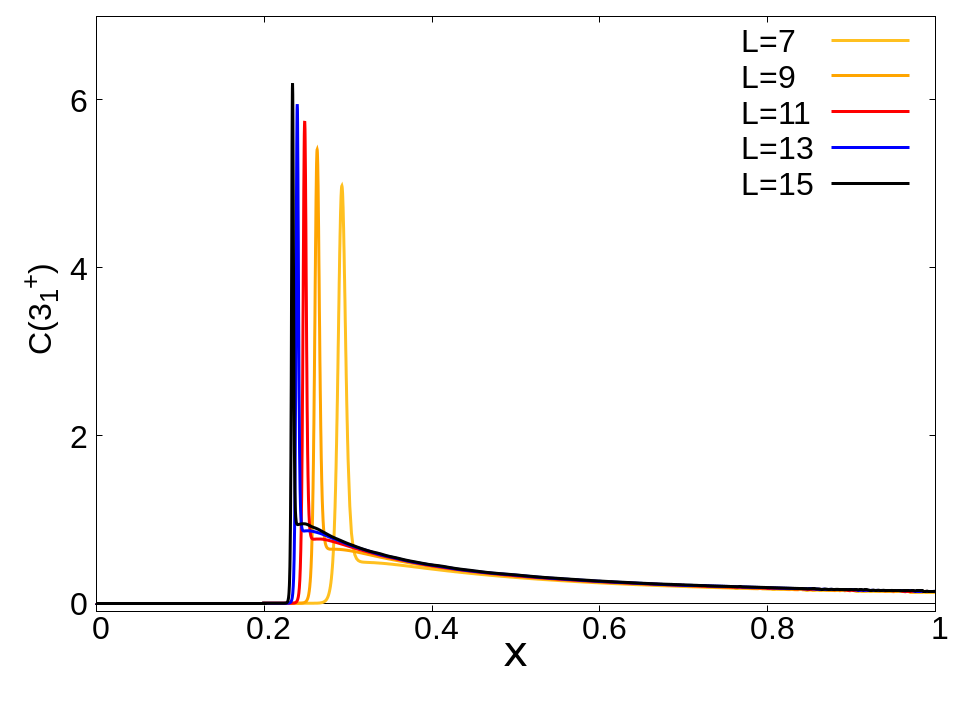}
	\caption{\textit{Finite size specific heat curves (see equation 
	\Ref{e20}) for the right handed trefoil.  For each value opf $L$ a sharp peak, 
	becoming more prominant as $L$ increases, develops in the vicinity 
	of the critical point.  The peaks (or ``spikes'') in the curves become
	narrower as $L$ increases, and increases in height with $L$.  Moreover, 
	with increasing $L$ the curves appear to converge to a limiting shape,
	and the approach of the curves to the spike from above develops a
	sharper turn when compared to the unknot in figure \ref{fC01}.}
	} 
\label{fC31}
\end{figure}

\begin{table}[h!]
\caption{Estimates of $x_{3_1}(L)$ and $H_{3_1}(L)$ for trefoils}
\def\arraystretch{1.33}
\begin{ruledtabular}
\begin{tabular}{c  @{}*{1}{c} @{}*{1}{c}  }         
$L$  & $x_{3_1}(L)$ & $H_{3_1}(L)$ \cr   
\hline
  $7$ & $0.2926 \pm 0.0055$ & $4.9760$  \cr
  $9$ & $0.2631 \pm 0.0035$ & $5.4079$  \cr
  $11$ & $0.2482 \pm 0.0025$ & $5.7470$ \cr
  $13$ & $0.2394 \pm 0.0018$ & $5.9504$ \cr
  $15$ & $0.2338 \pm 0.0015$ & $6.2013$ \cr
\end{tabular}
\end{ruledtabular}
\label{t2}   
\end{table}

In figure \ref{fC31} the specific heat curves are plotted for the trefoil.
these curves can be compared to the similar curves for the unknot in 
figure \ref{fC01}.  We note that there are again sharp peaks in the specific
heat, and in this case, narrower horizontally (and so more difficult to detect), 
and with significantly higher peaks.  Data on the locations of the peaks, and 
their heights, are listed in table \ref{t2}.

Plotting the estimates $x_{3_1}(L)$ as a function of $1/L^{1/\nu}$ shows
points that on a mild curve.  The critical point is extrapolated by using the
model
\begin{equation}
x_{3_1}(L) = x_{3_1} + a/L^{1/\nu} + b/L^{2/\nu} .
\label{e30}
\end{equation}
This gives
\begin{equation}
x_{3_1} = 0.2148 \pm 0.0069 .
\label{e31}
\end{equation}
This estimate is close to that of the unknot in equation \Ref{e26}.

Estimating the exponent $q$, using the same approach as before, now 
gives $q = -1.422 \pm 0.037$.  Plotting shows that this rescales the
heights of the peaks in the specific heat, but fails to rescale other
features in the graph appropriately.  This is shown in figure \ref{fC31scaled-alt}.
In particular, the specific heat curves for $x>x_{3_1}$ are systematically
separated in this plot.

\begin{figure}[h!]	
\includegraphics[width=0.5\textwidth]{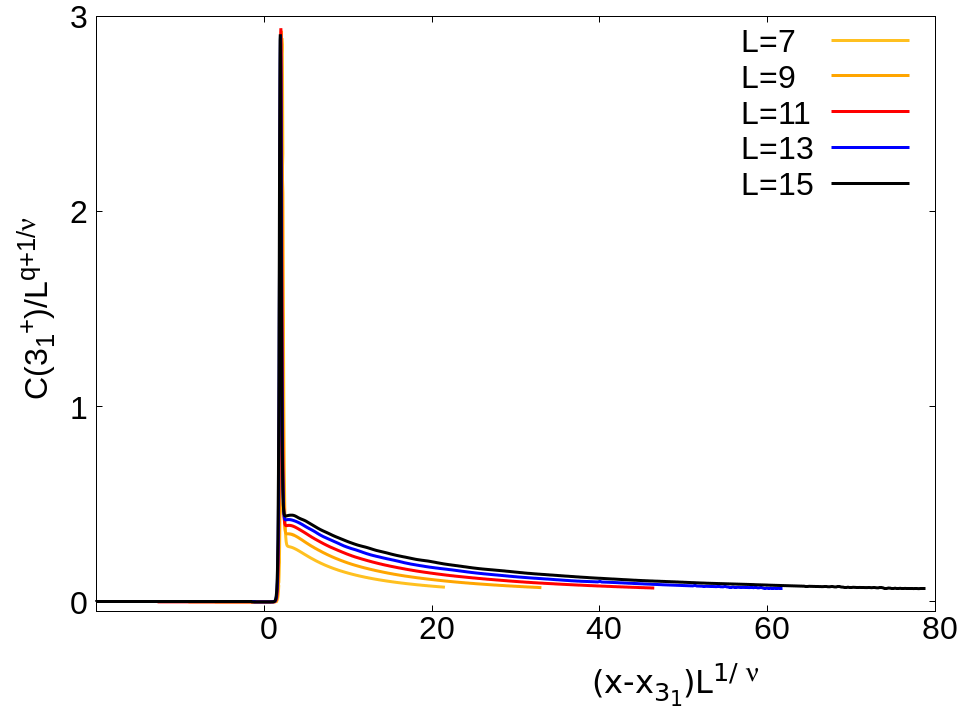}
	\caption{\textit{Rescaling of the specific heat curves of $3_1^{+}$
	using the value $q=-1.422$.  This value of $q$ clearly does not
	account for the scaling of the curves on the right hand side of the
	spike.}
	} 
\label{fC31scaled-alt}
\end{figure}

\begin{figure}[h!]	
\includegraphics[width=0.5\textwidth]{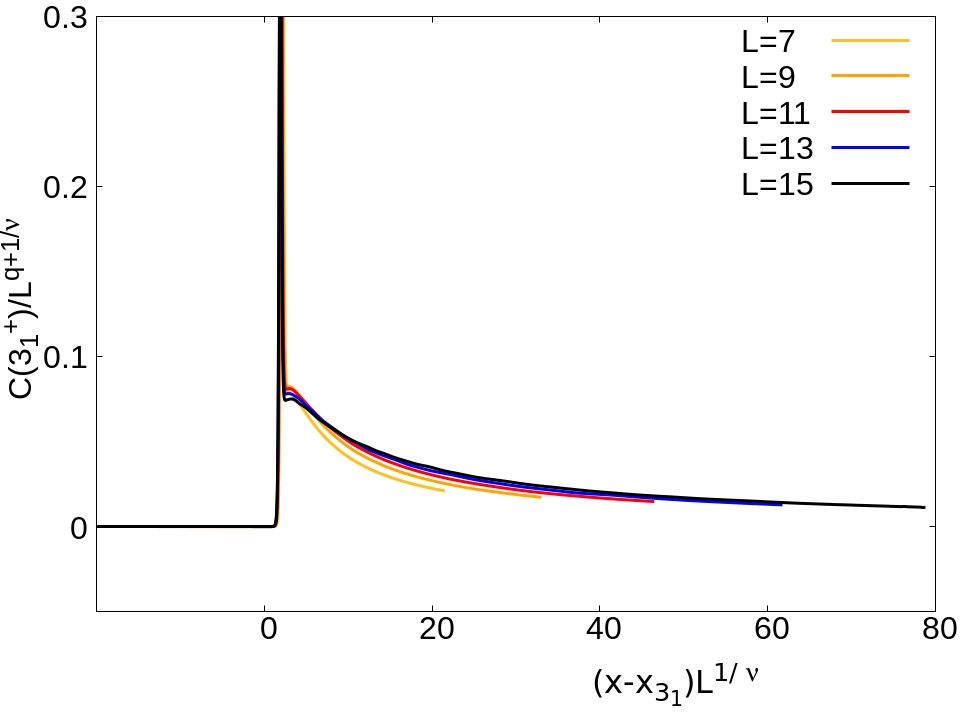}
	\caption{\textit{Rescaling of the specific heat curves of $3_1^{+}$ 
	using the unknot estimate $q=-0.765$ (equation \Ref{e28}).  This value
	appears to account for the scaling of the data.  In addition,
	a ``shoulder'' appears to develop as the curves approach the
	spike from above, a feature already visible in figures 
	\ref{fC31} and \ref{fC31scaled-alt}.}
	} 
\label{fC31scaled}
\end{figure}

An alternative possible value for the exponent $q$ is the unknot value
given in equation \Ref{e28}.  Rescaling the data for the trefoil using
this value gives the plot in figure \ref{fC31scaled}, where the unknot value
for $q$ rescales the specific heat curves above the critical point
appropriately.   This plot is numerical evidence that the exponent $q$, 
with its value determined from the unknot data, rescales the specific 
heat data determined for confined non-trivial lattice trefoils generally.
A consequence of this observation is that $\alpha^*_{3_1^+} 
= \alpha^*_{0_1}$.

In addition, a comparison of figure \ref{fC31scaled} with figure 
\ref{fC01scaled} (notice that the axes have the same scale) shows an 
apparent difference in the shape of the rescaled curves as they approach 
the critical point.  This is in addition to the (much) higher peaks in 
the specific heat for the trefoil, extending above the top boundary of 
the graph.  In figure \ref{Cscaled31} a schematic
drawing of the rescaled specific heat of confined lattice trefoils is shown.
As in figure \ref{Cscaled} the peaks are denoted by $p$, and the
junction point of the specific heat curves to the spike by $t$.  The rescaling
of the point $t$ in figure \ref{Cscaled31} is apparently controlled by
the exponent $q$ calculated from the data for unknotted lattice knots
(equation \Ref{e28}).  On the other hand, the height of the peaks scale with 
a different value, namely $q\approx -1.422$, as calculated directly from
the data in table \ref{t2}. In addition, the approach of the rescaled specific
heat from above, to its junction with the spike at $t$, as illustrated in 
figure \ref{Cscaled31}, shows additional changes when compared to the
case of the unknot, namely an inflection point and a concave, rather
than convex, shape on approach close the critical point.

The difference in the scaling of the peak heights in figures \ref{fC01} 
and \ref{fC31}) shows that the rescaling of the specific heat, as in
equation \Ref{scalingK}, is also a function of knot type (in that the
scaling function $h_K^\prime(z)$ is a function of $K$).  That is, it
appears that $h_{0_1}^\prime \not= h_{3_1^{+}}^\prime$, and
moreover, it may additionally be the case that the critical points
$x_{0_1}$ and $x_{3_1^{+}}$ may also be different, although this
cannot be ruled out, or shown, using our data.  One further notices
that above the critical points, but in the scaling region, as explained 
by figures \ref{Cscaled} and \ref{Cscaled31}, there are differences
in the shape of the rescaled specific heat curves of the unknot $0_1$
and right-handed trefoil $3_1^{+}$.  The developing spikes in the 
specific heat data of these knot types indicate a growing
non-analytic point in the free energy, with properties dependent on 
knot type, even if the scaling exponent $q$ turns out to be universal
to all knot types.

\begin{figure}[h!]
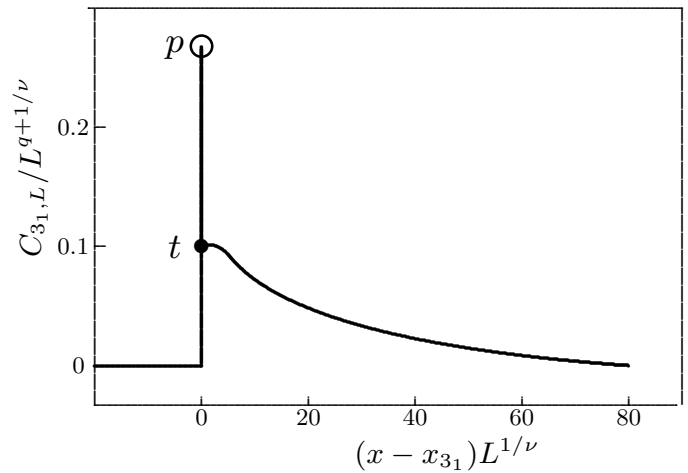

\beginpicture
\setcoordinatesystem units <2pt,1.5pt>
\setplotarea x from -10 to 100, y from -15 to 100

\put {\rotatebox{90}{\scalebox{1.25}{$C_{3_1,L}/L^{q+1/\nu}$}}} at -12 60 
\put {\scalebox{1.25}{$(x-x_{3_1})L^{1/\nu}$}} at 66 -12 

\plot 0 0 110 0 110 100 0 100 0 0 /

\multiput {\beginpicture \plot 0 0 0 2 / \endpicture} at 20 0 40 0 60 0 80 0 100 0  /
\put {$0$} at 20 -3 
\put {$20$} at 40 -3 
\put {$40$} at 60 -3 
\put {$60$} at 80 -3 
\put {$80$} at 100 -3

\multiput {\beginpicture \plot 0 0  2 0 / \endpicture} at 0 10 0 40 0 70  /
\put {$0$} at -3 10 
\put {$0.1$} at -4 40 
\put {$0.2$} at -4 70 

\setplotsymbol ({\scalebox{0.33}{$\bullet$}})
\plot 0 10 20 10 20 90 20 40 /
\setquadratic
\plot 25 38  50 20  100 10 /
\plot 20 40 22.5 40.2 25 38 /

\multiput {\scalebox{1.5}{$\bullet$}} at 20 40 /
\multiput {\scalebox{2.5}{$\circ$}} at 20 90 /

\put {\scalebox{1.5}{$p$}} at 15 90
\put {\scalebox{1.5}{$t$}} at 15 40
\endpicture

\caption{\textit{A schematic diagram emphasizing the features of the
	specific heat curves of the right handed trefoil in 
	figure \ref{Cscaled31}.  We notice that the peak heights (denoted ``$p$'')
	do not rescale consistent with the value of $q$ in equation \Ref{e28},
	but that the point denoted by $t$, at the junction of the curves
	with the spike, rescales consistent with the value in equation 
	\Ref{e28}.  In addition, a ``shoulder'' appears to develop near the 
	point $t$ in the curves incident with the spike.}
	} 
\label{Cscaled31}
\end{figure}

Rescaling the energy density similar to figure \ref{fE01scaled} gives
a similar outcome, but with somewhat larger corrections to scaling.
This is illustrated in figure \ref{fE31scaled}.  Notice that close to the
critical point, as the curves climb towards the turning point, they coincide,
but beyond the turning point they start to split apart as they exit the scaling
region near the critical point.

\begin{figure}[h!]	
\includegraphics[width=0.5\textwidth]{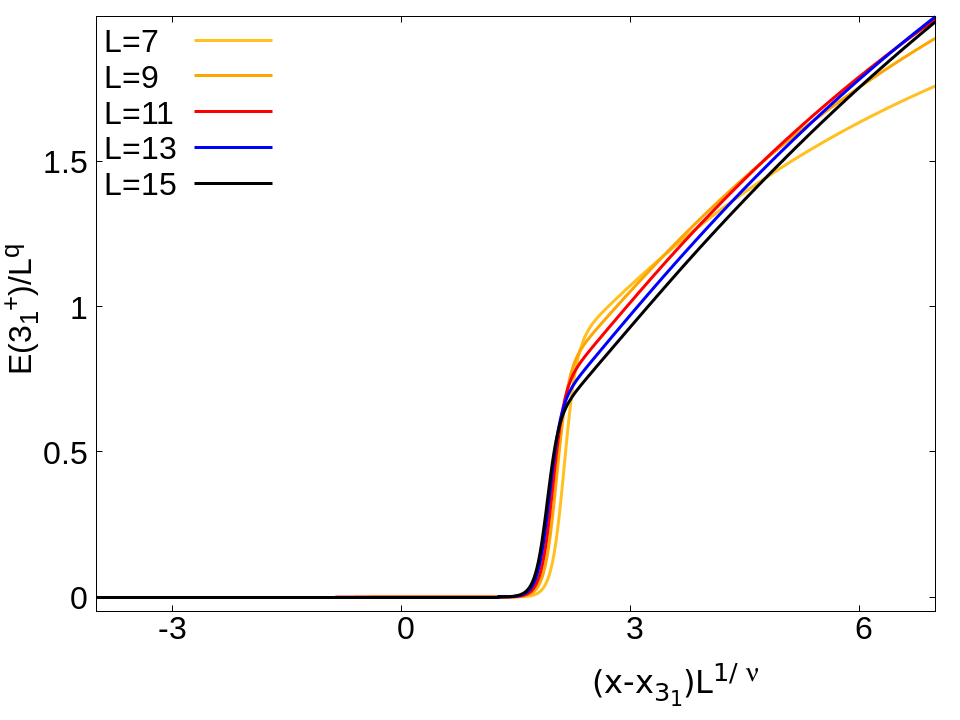}
	\caption{\textit{Rescaling the energy density curves (see equation \Ref{e22})
	of the right handed trefoil $3_1^{+}$.  Here, the value of $q$ is
	given by equation \Ref{e28}.  Compariton to the energy density 
	curves of the unknot (figure \ref{fE01scaled}) shows that, above the
	critical point, a sharper turn develops in the curves for the right handed
	trefoil knots.}
	} 
\label{fE31scaled}
\end{figure}

\subsubsection{The granny and square knots}

Data generated for the granny knot ($K=3_1^{+}\# 3_1^{+}$) and the
square knot ($K=3_1^{+}\#3_1^{-}$) were similarly collected and
analysed.  These are two compound knot types.  The granny knot
is a chiral knot (the right handed version is $3_1^{+}\#3_1^{+}$),
while the square knot is amphichireal.  The location and height of
the peaks in the specific heat of these knot types are listed in 
table \ref{t3}, and the rescaled specific heat curves are plotted in 
figure \ref{fC3131pppmscaled}, using $q=-0.765$ determined from the 
unknot data (equation \Ref{e28}).  The curves show features similar to
that seen for $K=3_1^{+}$ in figure \ref{fC31scaled}.  Comparison 
of these two knot types shows very similar data.

\begin{figure}[h!]	
\includegraphics[width=0.5\textwidth]{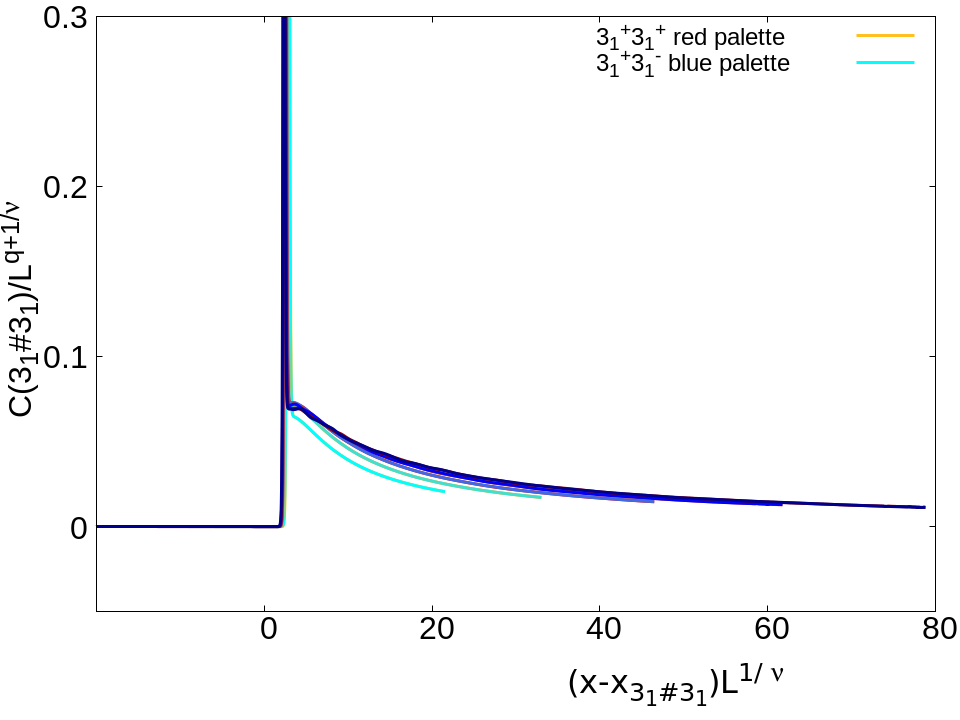}
	\caption{\textit{Rescaled specific heat curves for the granny
	$3_1^{+}\#3_1^{+}$ and square knot $3_1^{+}\#3_1^{-}$.  As 
	in figure \ref{fC31scaled}, a shoulder is developing at the junction
	of the curves with the spike, which is more prominant when compared
	to the spikes seen for the unknot and trefoil in figures \ref{fC01scaled} 
	and \ref{fC31scaled}.  The data for $3_1^{+}\#3_1^{+}$ are plotted
	in the blue, and for $3_1^{+}\#3_1^{-}$ in red.  Since the
	data are so similar for these knot types, only the blue curves are seen,
	laying on top of the red curves.}
	} 
\label{fC3131pppmscaled}
\end{figure}

\begin{table}[h!]
\caption{Estimates of $x_K$ and $H_K$ for $3_1^{+}\#3_1^{+}$ and $3_1^{+}\#3_1^{-}$}
\def\arraystretch{1.33}
\begin{ruledtabular}
\begin{tabular}{@{}*{1}{c} | @{}*{1}{c} @{}*{1}{c} |  @{}*{1}{c} @{}*{1}{c}  }         
$L$  & $x_{3_1^{+}\#3_1^{+}}(L)$ & $H_{3_1^{+}\#3_1^{+}}(L)$ 
        & $x_{3_1^{+}\#3_1^{-}}(L)$ & $H_{3_1^{+}\#3_1^{-}}(L)$ \cr   
\hline
  $7$ & $0.3189 \pm 0.0043$ & $9.0232$ & $0.3176 \pm 0.0043$ & $8.9088$  \cr
  $9$ & $0.2774 \pm 0.0026$ & $9.9477$ & $0.2765 \pm 0.0026$ & $9.7852$  \cr
  $11$ & $0.2574 \pm 0.0018$ & $10.5902$ & $0.2567 \pm 0.0018$ & $10.3658$ \cr
  $13$ & $0.2460 \pm 0.0014$ & $11.0789$ & $0.2455 \pm 0.0014$ & $10.8330$ \cr
  $15$ & $0.2387 \pm 0.0011$ & $11.4350$ & $0.2382 \pm 0.0011$ & $11.1553$  \cr
\end{tabular}
\end{ruledtabular}
\label{t3}   
\end{table}

Using a model similar to equation \Ref{e30} to extrapolate the critical
point in each of these cases we obtain the estimates
\begin{eqnarray}
x_{3_1^{+}\#3_1^{+}} &=& 0.2152 \pm 0.0052; \\
x_{3_1^{+}\#3_1^{-}} &=& 0.2149 \pm 0.0052.
\end{eqnarray}
These values are very close to the estimates of $x_{0_1}$ and $x_{3_1^{+}}$
(the critical points for the unknot and the trefoil) -- see equations \Ref{e26}
and \Ref{e31}.  For finite values of $L$, however, the  estimates are well 
separated.  For example, in figure \ref{fC313131scaled} the rescaled 
specific heat curves for the knot types 
$\{0_1,3_1^{+},3_1^{+}\#3_1^{+},3_1^{+}\#3_1^{-}\}$
are plotted in a region close to the peaks for $L=15$.  In these plots
the horizontal direction ($x$-axis) is magnified (compare the scales with that in 
figure \ref{fC3131pppmscaled}).  Observe that the spikes (now more properly 
called peaks)  for the compound knot types $3_1^{+}\#3_1^{+}$ 
and $3_1^{+}\#3_1^{-}$ 
are virtually identical on this scale, but they are also well separated from the
data for the trefoil $3_1^{+}$ and the unknot $0_1$.  Using the width of the peaks
as an estimate of the uncertainty in the location of the critical points, show that
the critical points of the square and granny knot types here are well outside the
confidence intervals of the the prime knot types $3_1^{+}$ and $0_1$.
However, when extrapolating the critical points to the thermodynamic limit,
one obtains critical values which are well within the stated confidence intervals
of each other.

\begin{figure}[h!]	
\includegraphics[width=0.5\textwidth]{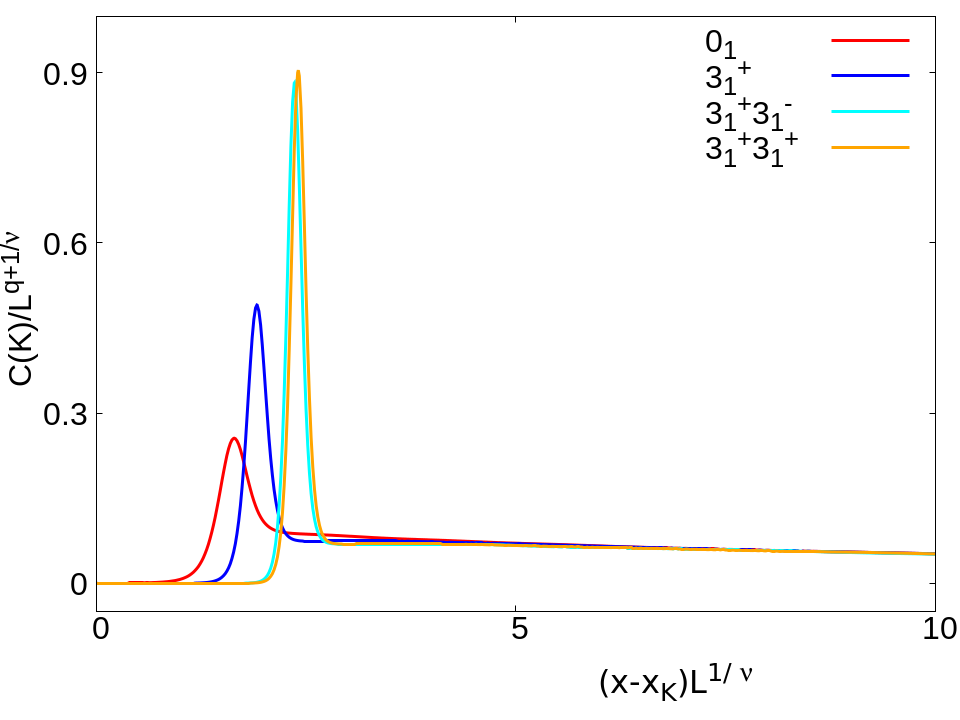}
	\caption{\textit{Rescaled specific heat curves for $0_1$ (red), $3_1^{+}$
	(blue), $3_1^{+}\#3_1^{+}$ (orange) and $3_1^{+}\#3_1^{-}$ (cyan).
	In this figure horizontal direction is greatly magnified (compare the scale to
	that of figure \ref{fC3131pppmscaled}), and the spikes seen in other
	graphs now appears as peaks in the data.  On this scale, the peaks in
	$3_1^{+}\#3_1^{+}$ and $3_1^{+}\#3_1^{-}$ are virtually identical 
	and well separated from the data for $0_1$ and $3_1^{+}$.  Notice that
	the peaks become more prominent (higher), and narrower, as the knot types
	varies from $0_1$, $3_1^{+}$, and then to the compound knot types.}
	} 
\label{fC313131scaled}
\end{figure}

\subsubsection{Other prime knots}

In addition similar data were obtained for prime knot types up to 
six crossings.  These are the knot types 
$\{4_1,5_1^{+},5_2^{+},6_1^{+},6_2^{+},6_3\}$.
The estimates of the critical point and peak heights for the figure
eight knot ($4_1$) are listed in table \ref{t4}.  Extrapolating
the critical point using the model in equation \Ref{e30} gives
\begin{equation}
x_{4_1} = 0.2148 \pm 0.0061 .
\end{equation}

\begin{table}[h!]
\caption{Estimates of $x_{4_1}(L)$ and $H_{4_1}(L)$}
\def\arraystretch{1.33}
\begin{ruledtabular}
\begin{tabular}{c  @{}*{1}{c} @{}*{1}{c}  }         
$L$  & $x_{4_1}(L)$ & $H_{4_1}(L)$ \cr   
\hline
  $7$ & $0.3000 \pm 0.0049$ & $6.2863$  \cr
  $9$ & $0.2671 \pm 0.0031$ & $6.8413$  \cr
  $11$ & $0.2507 \pm 0.0022$ & $7.2202$ \cr
  $13$ & $0.2415 \pm 0.0016$ & $7.6844$ \cr
  $15$ & $0.2351 \pm 0.0013$ & $7.7725$ \cr
\end{tabular}
\end{ruledtabular}
\label{t4}   
\end{table}

By equation \Ref{e27} one expects that  
\begin{equation}
H_K(L) = \mathcal{C}_{K,L}(x_K) \simeq L^{q+1/\nu} h_K^\prime(0) .
\end{equation}
Thus, the value of the scaling function at the critical point,
$h_K^\prime (0)$ (see equation \Ref{e27}) can be defined by the
limit 
\begin{equation}
\lim_{L\to\infty} \frac{H_K(L)}{L^{q+1/\nu}} = h_K^\prime (0) .
\end{equation}
For finite values of $L$ one expects that
\begin{equation}
\frac{H_K(L)}{L^{q+1/\nu}} = h_K^\prime (0) + 
\frac{a}{L^{q+1/\nu}} 
+ \cdots .
\label{38}
\end{equation}
In table \ref{t5} the heights of the peaks in $C_{K,L}(x_K(L)) = H_K(L)$
are given, collecting the data from tables \ref{t1}, \ref{t2}, \ref{t3} and \ref{t4},
and the adding data for the other prime knots considered.

\begin{table}[h!]
\caption{Estimates of $H_{K}(L)$}
\def\arraystretch{1.33}
\begin{ruledtabular}
\begin{tabular}{l  @{}*{5}{c} }
  & \multicolumn{5}{c}{$L$} \cr         
   \cline{2-6} 
$K$  & $7$ & $9$ & $11$ & $13$ & $15$ \cr   
\hline
  $0_1$ & $1.5763$ & $2.0141$ & $2.4509$ & $2.8452$ & $3.2331$  \cr
  $3_1^{+}$ & $4.9760$ & $5.4079$ & $5.7470$ & $5.9594$ & $6.2013$  \cr
  $4_1$ & $6.2863$ & $6.8413$ & $7.2202$ & $7.6844$ & $7.7725$ \cr 
  $5_1^{+}$ & $7.5451$ & $8.2829$ & $8.8108$ & $9.1892$ & $9.5177$ \cr 
  $5_2^{+}$ & $7.6131$ & $8.3224$ & $8.8107$ & $9.2183$ & $9.5450$ \cr 
  $6_1^{+}$ & $8.8670$ & $9.7451$ & $10.3373$ & $10.8385$ & $11.2580$ \cr 
  $6_2^{+}$ & $8.9765$ & $9.8726$ & $10.4842$ & $10.9422$ & $11.3070 $ \cr 
  $6_3$ & $9.0197$ & $9.9182$ & $10.4933$ & $10.9905$ & $11.3676$ \cr 
  $3_1^{+}\#3_1^{+}$ & $9.0232$ & $9.9477$ & $10.5902$ & $11.0789$ & $11.4350$ \cr 
  $3_1^{+}\#3_1^{-}$ & $8.9088$ & $9.7852$ & $10.3658$ & $10.8330$ & $11.1553$ \cr 
\end{tabular}
\end{ruledtabular}
\label{t5}   
\end{table}

Ignoring higher order terms in equation \Ref{38} and estimating
$h_K^\prime (0)$ using linear fits give the extrapolated values
in the last column of table \ref{t6}. The value for the unknot is
$h_{0_1}^\prime(0) \approx 0.259$,  and this is larger than the
estimate for the trefoil $3_1^{+}$.  Increasing knot complexity
down the table generally increases, with some variations, the
value of $h_K^\prime (0)$.  Also observe that the estimates for
$5$-crossing and $6$-crossing knots are close in value to each 
other.

\begin{table}[h!]
\caption{Extrapolated estimates of $x_{K}$}
\def\arraystretch{1.33}
\begin{ruledtabular}
\begin{tabular}{l  @{}*{2}{c}  }         
$K$  & $x_{K}$ & $h_K^\prime(0)$ \cr   
\hline
  $0_1$ & $0.2161 \pm 0.0069$ & $0.259$ \cr
  $3_1^{+}$ & $0.2148 \pm 0.0069$ & $0.197$ \cr
  $4_1$ & $0.2148 \pm 0.0061$ & $0.253$ \cr
  $5_1^{+}$ & $0.2146 \pm 0.0056$ & $0.322$ \cr
  $5_2^{+}$ & $0.2148 \pm 0.0056$ & $0.314$ \cr
  $6_1^{+}$ & $0.2154 \pm 0.0052$ & $0.387$ \cr
  $6_2^{+}$ & $0.2150 \pm 0.0052$ & $0.382$ \cr
  $6_3$ & $0.2152 \pm 0.0052$ & $0.383$ \cr
  $3_1^{+}\#3_1^{+}$ & $0.2152 \pm 0.0052$ & $0.398$ \cr 
  $3_1^{+}\#3_1^{-}$ & $0.2149 \pm 0.0052$ & $0.371$ \cr
\end{tabular}
\end{ruledtabular}
\label{t6}   
\end{table}

Additionally, the estimated critical points are listed
in the middle column of table \ref{t6}.  These values
are very close to each other.  Assuming that all these
knot types have the same critical point, one can calculate
their average, which is $0.2151 \pm 0.0014$.
This estimate is slightly smaller that the average over
the knot types in table \ref{t6}.  This is to be expected,
since the connective constant for unknots are known
to be strictly less than that of all polygons \cite{SW88,P89}.  For
non-trivial knot types it is conjectured that the connective
constant is equal to that of the unknot (and it is known
that it is equal or larger than that of the unknot).

\section{Discussion}

Our data, as represented by the graphs of the free energy, the 
finite size energy density, and the specific heat data in the last section, 
indicate sharp transitions at a critical point $x_K$ for each of the
knot types examined.  This transition separates a \textit{solvent phase}
(when $x<x_K$ and illustrated schematically in figure \ref{1}) 
dominated by solvent, from a \textit{polymer phase} (as schematically
illustrated in figure \ref{2}) where the cube has a positive density of 
polymer (as, for example, seen in figure \ref{fE01}).  
Data analysis indicates that knot type may affect the transition 
behavior. This is most apparent in the varying heights of specific 
heat spikes, and more subtly in the shape of the specific heat curves 
as the system approaches the critical point in the polymer phase.

With the above in mind, we now turn our attention to the question 
raised by the caption of figure \ref{2}.  That is, is the knot captured
in a small knotted arc in the polymer phase, or does it relax (loosens)
or ``dissolve'' in this phase?  To examine this question, consider the 
ratio of energy densities $\mathcal{E}_{K_1,L} (x)/\mathcal{E}_{K_2,L} (x)$. 
By equation
\Ref{e22},
\begin{equation}
\frac{\mathcal{E}_{K_1,L} (x)}{\mathcal{E}_{K_2,L} (x)}
\sim \frac{h_1 ((x{-}x_{K_1})\,L^{1/\nu}) }{h_2 ((x{-}x_{K_2})\,L^{1/\nu}) }
\end{equation}
assuming that the exponent $q$ is not a function of knot type 
(as seen in the previous section), and where $h_j$ is the scaling
function associated with knot type $K_j$ with a critical point $x_{K_j}$. 

Consider first the case $K_2 = 0_1$ (the unknot) and $K_1 = 3_1^+$
(the trefoil).  In the solvent phase $x < \min\{x_{0_1},x_{3_1}\}$ and the
length $n$ of lattice knots are small.  Our data shows that $p_{n,L}(0_1) 
\gg p_{n,L}(3_1)$ in this regime with the result that
$(\mathcal{E}_{3_1^+,L} (x)/ \mathcal{E}_{0_1,L} (x)) \approx 0$.
In the polymer phase (when $x>\max\{x_{0_1},x_{3_1}\}$) one may
examine the asymptotic scaling for large $L$, suggested by equation 
\Ref{e19}.  This shows that, if $\alpha^*_{0_1}=\alpha^*_{3_1^+}$, 
the ratio $(\mathcal{E}_{3_1^+,L} (x) / \mathcal{E}_{0_1,L} (x))$
should be constant, independent on $x$.  On the other hand, if
$\alpha^*_{0_1}\not=\alpha^*_{3_1^+}$, then the ratio should be 
a function of $x$.

In figure \ref{fR3101} we plot 
$(\mathcal{E}_{3_1^+,L} (x) / \mathcal{E}_{0_1,L} (x))$
as a function of $x$.  In the solvent phase the ratio is,  as expected, 
approximately zero, but it increases sharply through the critical
point(s), forming a minor spike, before settling down close to
$1$ in the polymer phase.  This suggests that
\begin{equation}
\frac{\mathcal{E}_{3_1,L} (x)}{\mathcal{E}_{0_1,L} (x)}
\simeq 
\begin{cases}
0, & \text{if $x< \min\{x_{0_1},x_{3_1}\}$ (solvent)}; \\
1, & \text{if $x> \max\{x_{0_1},x_{3_1}\}$ (polymer)}. \\
\end{cases}
\end{equation}
In the vicinity of the critical point the minor spike shows that the
ratio exceeds $1$, and that trefoils are dominant over the unknot
in this region.  This is consistent with the trefoil being localized in
this regime, in line with the expectation of equation \Ref{e11}.
For larger values of $x$ the ratio settles down on $1$.
This indicates that the exponent $\alpha^*_K$ may be 
independent of knot type.  This is consistent with the knot transitioning 
from a tightly packed ball-pair structure in the solvent phase to become 
a looser topological constraint on the ring in the polymer phase, where 
it no longer affects thermodynamic properties.  In other words, if the 
solvent phase is schematically represented in the left panel of 
figure \ref{1},  then the polymer phase is more appropriately represented 
by the right panel of Figure \ref{2}. These arguments are also consistent
with the observation earlier, that $\alpha_{3_1^+}^* = \alpha_{0_1}^*$
in equation \Ref{e19}, since the same value of $q$ (see equation
\Ref{e24}) rescales both the unknot and trefoil knot data.

\begin{figure}[h!]	
\includegraphics[width=0.5\textwidth]{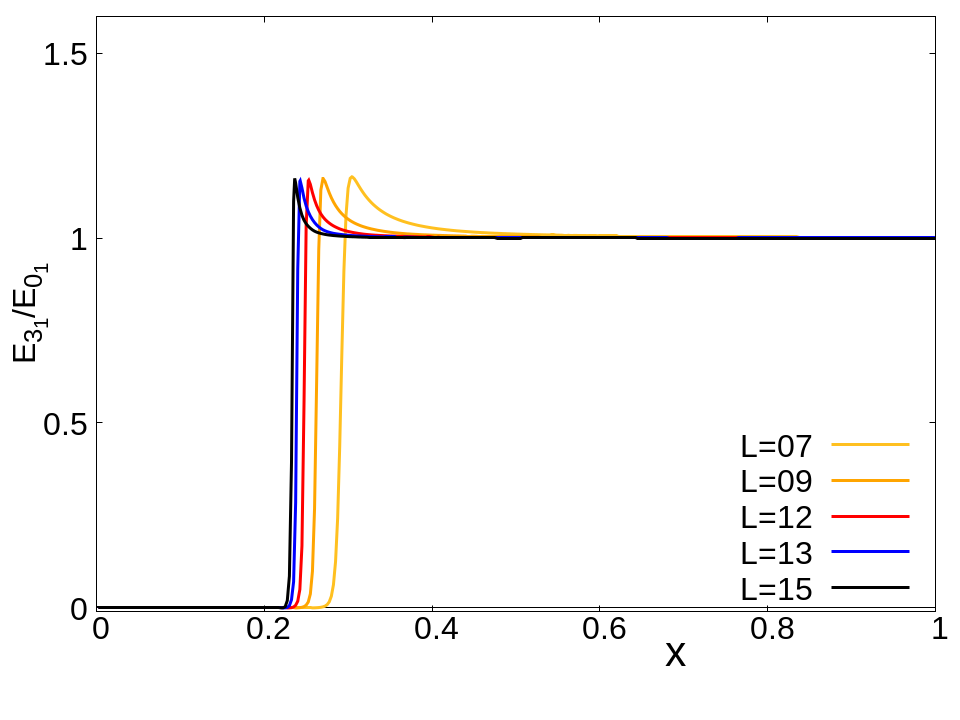}
	\caption{\textit{The energy density ratio of trefoils ($3_1$) to 
	unknots ($0_1$) 	for confining boxes of sizes $\{7,9,11,13,15\}$.  With 
	increasing $L$ the curves appear to accumulate on a step-function 
	with a minor spike.  At small $x$ the ratio is approximately zero,
	and with increasing $x$ passes through a step and spike to
	be approximately equal to $1$.  This behaviour is consistent with 
	the relaxation of melting of the knot from a tight conformation
	in the solvent phase, to a looser topological constraint on
	the conformations of the knot in the polymer phase.}
	} 
\label{fR3101}
\end{figure}

Similar results are found for other knot types.  In figure \ref{fR4101}
the ratio of the figure eight knot ($4_1$) and the unknot is plotted.
The general observations for this case are similar to those seen
in figure \ref{fR3101} -- a sharp transition from very small values,
through a spike, and the data levels off close to $1$ as $x$ increases.
The ratio for the figure eight knot to the trefoil is similarly plotted
in figure \ref{fR4131}.  In this case the foot of the step is somewhat
rounded compared to figures \ref{fR3101} and \ref{fR4101}, and the
spike is nearly non-existent at the top of the step.  The graphs then
accumulate on $1$ with increasing $L$ in the polymer phase.

\begin{figure}[h!]	
\includegraphics[width=0.5\textwidth]{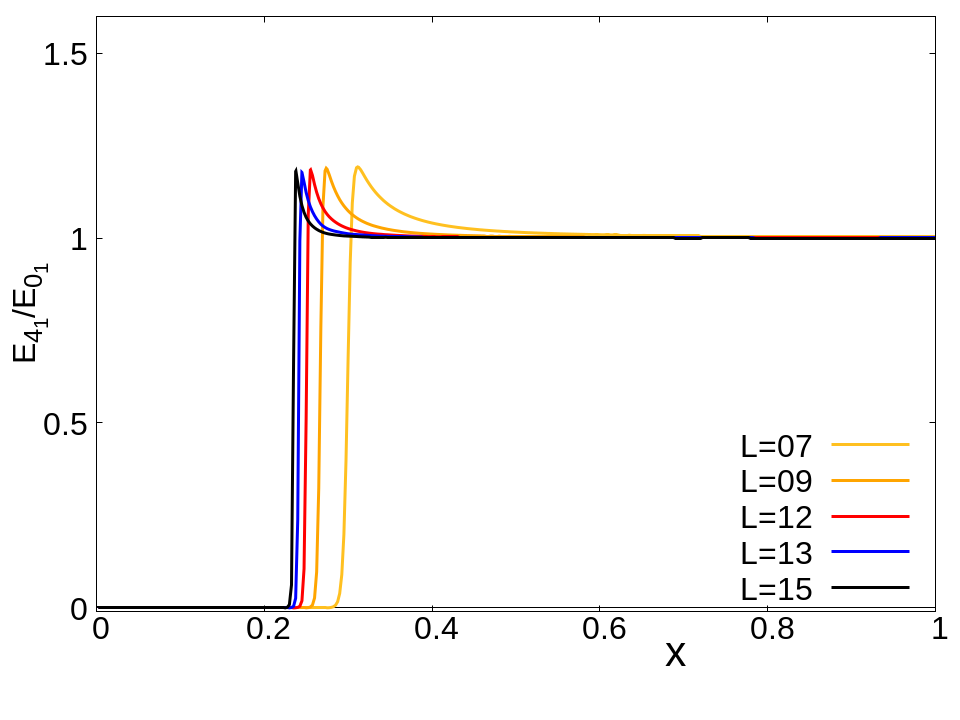}
	\caption{\textit{The energy density ratio of figure eights ($4_1$)
	to unknots ($0_1$) for confining boxes of sizes $\{7,9,11,13,15\}$.  
	The general features in this graph is very similar to that seen 
	in figure \ref{fR3101}.}
	} 
\label{fR4101}
\end{figure}

As a last example, we plotted a compound knot type (the granny
knot $3_1^+\# 3_1^+$) over the unknot in figure \ref{fR31p31p01}.
In this plot the features seen in figure \ref{fR3101} are recovered,
and even though there are prime factors in the granny knot, the
data again accumulate on $1$ in the polymer phase with increasing
size $L$ of the confining box.

\begin{figure}[h!]	
\includegraphics[width=0.5\textwidth]{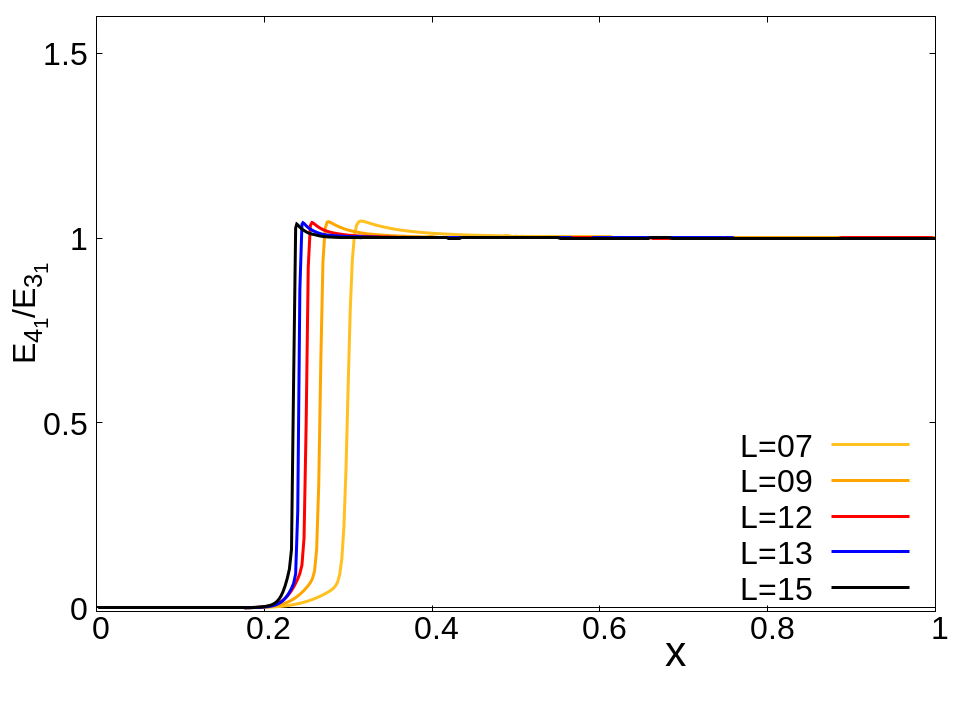}
	\caption{\textit{The energy density ratio of figure eights ($4_1$)
	to trefoils ($3_1$) for confining boxes of sizes $\{7,9,11,13,15\}$.  
	The general features in this graph is very similar to that seen 
	in figure \ref{fR3101}.  However, notice the slight rounding
	of the curves at the foot of the step, and the significant reduction
	in the height of the spike seen in figures \ref{fR3101} and 
	\ref{fR4101}.}
	} 
\label{fR4131}
\end{figure}

\begin{figure}[h!]	
\includegraphics[width=0.5\textwidth]{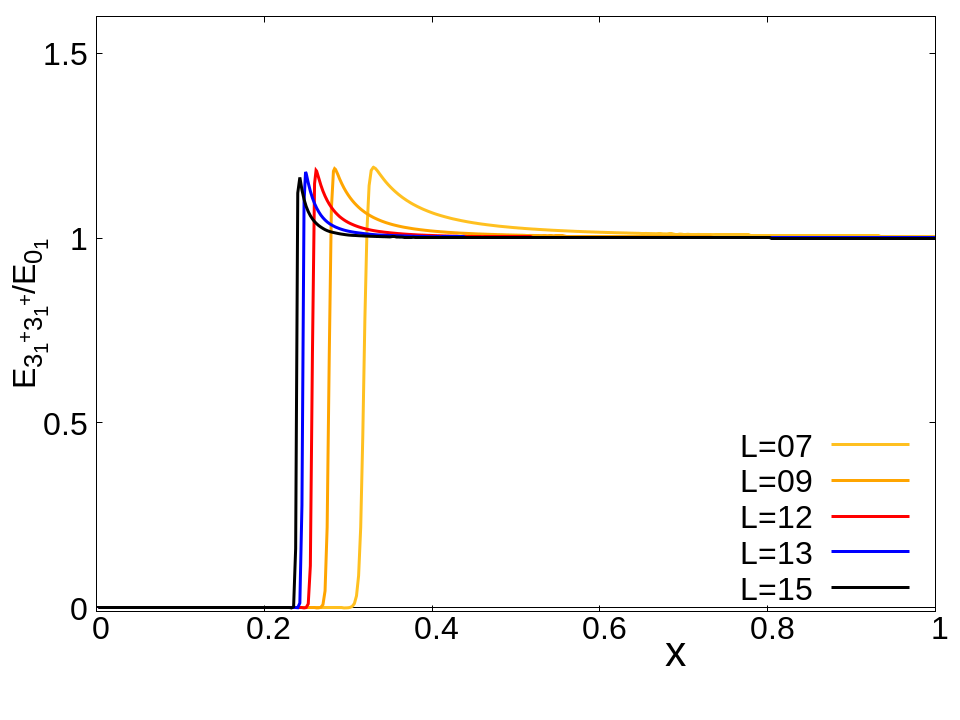}
	\caption{\textit{The energy density ratio of the rings of compound
	knot type ($3_1^{+}\#3_1^{+}$) 	to unknots ($0_1$) for confining boxes 
	of sizes $\{7,9,11,13,15\}$. The general features in this graph is very 
	similar to that seen in figure \ref{fR3101}.}
	} 
\label{fR31p31p01}
\end{figure}

Finally we briefly examine our results for the free energy in terms
of Flory-Huggins theory. The finite size free energy density 
$f_{K,L}(x)$ is defined in equation \Ref{e14}.  The mean concentration 
of polymer in the cube, $\phi_{K,L}(x) = \mathcal{E}_{K,L}(x)$, 
is approximately given by equation \Ref{e21}.  We notice that the free energy 
\textit{per unit length polymer} is given by
\begin{equation}
\psi_{K,L}(x) = f_{K,L}(x)/\phi_{K,L}(x) .
\label{eqn41}
\end{equation}
Flory-Huggins theory \cite{Flory42,Huggins42,F69} gives a 
mean field expression for the \textit{excess} free energy $F$ by
\begin{equation}
F (\phi) = (1/V) \log \phi + (1-\phi)\log(1-\phi) - \chi \phi^2
\end{equation}
where $\chi$ is the Flory Interaction Parameter and $V=L^3$ is the
volume of the confining cube.  In the case of lattice self-avoiding
walks and lattice knots this was examined in 
references \cite{GJvR18,JvR19,JvR19a}.

One may apply Flory-Huggins theory to our results by noting that the
excess free energy is given by $\psi_e(x) = \psi_{K,L}(x) - \log \mu$ (see 
equation \Ref{eqn41}), where $\log \mu$ is the background conformational free 
energy.  This excess free energy can be plotted against $\phi_{K,L}(x)$, which 
is the mean concentration or density of polymer (equation \Ref{e21}). This should 
collapse the excess free energy $\psi_e(x)$ to a single underlying curve for all 
the values of $L$ we considered.  This is shown in figure \ref{floryhuggins}
for the unknot free energies. Similar graphs can be made for 
the other knot types.  We notice, as seen in references \cite{GJvR18,JvR19,JvR19a},
a very clear separation into a solvent phase for small concentration.
and a polymer phase for large concentration.

\begin{figure}[h!]	
\includegraphics[width=0.5\textwidth]{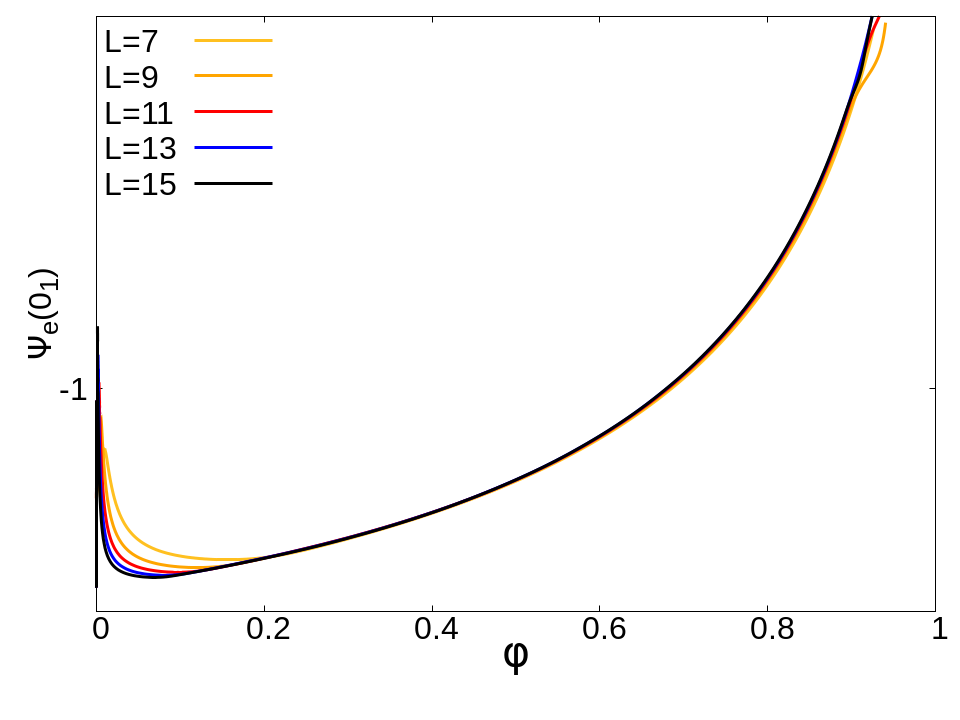}
	\caption{\textit{The excess free energy $\psi_e(x) = \psi_{K,L}(x) - \log \mu$
	plotted as a function of the mean concentration $\phi_{K,L}(x)$.  
	According to Flory-Huggins theory these curves should coincide with
	the mean field free excess free energy.  In this graph the curves collapse
	for all the values of $L$ considered to an underlying curve, except near 
	small concentrations where finite size effects are more pronounced.}
	} 
\label{floryhuggins}
\end{figure}

%

-----------------------

\section*{Acknowledgements} EJJvR acknowledges financial support 
from NSERC (Canada) in the form of Discovery Grant RGPIN-2019-06303. 
EJJvR was also grateful to the Department of Physics and Astronomy at
the University of Padova for financial support during a visit in 2019. 
Data generated for this research are available at Zenodo \cite{JvRTO26data}.

\bibliographystyle{unsrt}
\bibliography{knotbox}

\end{document}